\documentclass[12pt]{amsart}
\usepackage[utf8]{inputenc}
\usepackage{amssymb,amsmath,amsfonts,latexsym,verbatim}
\usepackage[dvipsnames]{xcolor}

\newcommand*\diff{\mathop{}\!\mathrm{d}}

\def\D{\mathcal{D}}

\def\U{\mathcal{U}}

\def\Re{\mathbf{R}}

\def\ep{\varepsilon}
 
\def\ta{\theta}

\def\la{\lambda}
\def\da{\delta}
 
\def\phi{\varphi}

\def\ul{\underline}

\newcommand{\df}[1]{\textit{\textbf{#1}}}

\newcommand{\abs}[1]{ \left | #1 \right | }

\def\exb{\mathbf{x}}

\usepackage{tikz}
\usetikzlibrary{trees}
\usetikzlibrary{patterns}

\usepackage{natbib}

	\usepackage[left=1.5in,top=1.25in,right=1.5in,bottom=1.25in]{geometry}
	\usepackage[onehalfspacing]{setspace}

\title[Dynamically inconsistent]{Revealed preferences for dynamically inconsistent models}
\author[Echenique and Tserenjigmid]{Federico Echenique and Gerelt Tserenjigmid}
\thanks{Echenique: Department of Economics, UC Berkeley, 
\texttt{fede@econ.berkeley.edu} \\ Tserenjigmid: Department of Economics, UC Santa Cruz,  \texttt{gtserenj@ucsc.edu}. We are grateful to Laura Blow for comments on an earlier draft.}

\newtheorem{theorem}{Theorem}

\newtheorem{thm}[theorem]{Theorem}
\newtheorem{proposition}[theorem]{Proposition}

\newtheorem{lem}[theorem]{Lemma}

\theoremstyle{definition}

\newtheorem{defn}{Definition}

\theoremstyle{remark}

\newtheorem{remark}{Remark}

\theoremstyle{example}

\newtheoremstyle{named}{}{}{\itshape}{}{\bfseries}{:}{.5em}{#1\thmnote{#3}}
\theoremstyle{named}
\newtheorem*{namedaxiom}{}

\begin{document}

\maketitle

\begin{abstract}
    We study the testable implications of models of dynamically inconsistent
    choices when planned choices are unobservable, and thus only ``on path'' data is available. First, we discuss the approach in 
    \cite*{blowbrowningcrawford2021}, who characterize first-order
    rationalizability of the model of quasi-hyperbolic discounting. We show that the first-order approach does not guarantee rationalizability by means of the quasi-hyperbolic model. This motivates consideration of an abstract model of intertemporal choice, under which we provide a characterization of different behavioral models -- including the naive and sophisticated paradigms of dynamically inconsistent choice.
\end{abstract}

\section{Introduction}
A dynamically inconsistent decision-maker will make plans that she later does not carry out. If her plan is not carried out, it cannot be observed from the decision-maker's actions. This non-observability of inconsistent plans presents a unique challenge when we want to understand the testable implications of dynamically inconsistent models, because it means that the most obvious implication of dynamic inconsistency (its defining characteristic) is inherently unobservable. Of course, in experimental settings, one may imagine eliciting the decision-maker's plans, but such elicitation is in general very challenging, and certainly not possible with the non-experimental, observational, data that is often used in empirical work. Our purpose in the present paper is to discuss, and ultimately characterize, the testable implications of dynamically inconsistent model when plans are unobservable.

Imagine a decision-maker, Alice, who has to make choices in periods one and two: this week and next. She may have certain preferences over the sequence of choices, perhaps how much to consume in each week; but her preferences in the second week may be governed by different preferences than she holds initially over sequences of choices throughout the two-week span. Her week-two self may then deviate from her planned choices in week one. Such behavior is called dynamically inconsistent: the notion of dynamic inconsistency requires comparing planned to actual behavior. The problem, however, is that if we only have access to choice data, then Alice's plan is unobservable because her planned choices are not carried out. How can we then determine that Alice's week-two self deviated from her initial week-one plan?

The literature often assumes that Alice's two selves are engaged in a non-cooperative game, and that a game-theoretic equilibrium outcome is determined by backward induction. Alice's week-two preferences determine an optimal choice, for each choice of her week-one self. So in week one she chooses an optimal action, understanding that in week two she will choose optimally given her second-week preferences. The equilibrium outcome is what we call \df{on-path data}. It contrasts with the possibility of observing a plan for Alice's first-week preferences, or observing a full contingent strategy for her second-week self. We focus on settings where such planned ``counterfactual'' evidence is unavailable, and assume that one can only observe on-path behavior.  

Our paper proceeds in two parts. First we assess the main existing effort to  characterize the testable implications of dynamically inconsistent behavior using on-path data. We find that this problem is very challenging, and discuss the proposed solutions. In second place, we consider general preferences in an abstract model of choice, with multiple on-path observations. Here we are able to characterize the testable implications of equilibrium behavior. The combination of the lack of structure, and multiple observations, renders the problem tractable.

Our starting point is \cite{blowbrowningcrawford2021}, the main existing effort to study dynamic inconsistency and revealed preferences, who present a revealed-preference characterization of quasi-hyperbolic discounting preferences in a demand-theory setting.  A consumer chooses over time, according to intertemporal utility tradeoffs that change over time because of the quasi-hyperbolic assumption. Following \cite{afriat1967construction} (in the general utility-maximization framework), and \cite{browning1989nonparametric} (for dynamically consistent intertemporal choice), the authors use a \df{first-order approach}. That is, \cite{blowbrowningcrawford2021} interpret the consistency of data with the model as \emph{the existence of a solution to a system of equations that captures the first-order conditions for utility maximization}. In the case of dynamically inconsistent quasi-hyperbolic discounting agents, these first-order conditions come from Euler equations derived in, for example, \cite{harris2001dynamic}.

Now, in the case of Afriat and Browning, one can show that the first-order approach is equivalent to saying that there is an instance of the model that explains the data. Indeed, in the general model of utility maximization, and in the model of a dynamically consistent consumer with exponential discounting, a dataset is consistent with the first-order approach if and only if one can find a utility function that satisfies the conditions laid out in the model, and that generates the data as optimal choices.  Actually, proving this fact is, arguably, Afriat's main contribution. 

We show that the no-loss-of-generality of the first-order approach does not hold for the model of a dynamically inconsistent quasi-hyperbolic agent. In other words, the first-order approach is too permissive. To explain, let us say that a dataset is \df{FOCs rationalizable} if there are model parameters such that the system of inequalities in \cite{blowbrowningcrawford2021} is satisfied. We also say that a dataset is \df{equilibrium rationalizable} if there are model parameters such that observed consumption is an equilibrium outcome of the quasi-hyperbolic discounting model. Our Theorem~\ref{thm:example} shows that there are datasets that are FOCs rationalizable, but not equilibrium rationalizable. 

The conclusion in our Theorem~\ref{thm:example} matters, beyond its theoretical implications, because Blow et al.\ carry out an empirical application in which they emphasize the added empirical explanatory power of the quasi-hyperbolic model. The gap between \emph{equilibrium rationalizability} and \emph{FOCs rationalizability} may call into question the conclusion about the explanatory power of the quasi-hyperbolic model. 

We should emphasize that the problem is difficult, and \cite{blowbrowningcrawford2021} make clear progress. Our Theorem~\ref{thm:example} points out issues with the first-order approach, but we are unable to provide a characterization of the testable implications in the one-observation consumption setting of Blow et.\ al. The first-order approach is tractable, even if subject to the critiques of our paper.

In search for a tractable formulation, we turn to a general and abstract model of choice with dynamic inconsistency. We consider a dataset of choices from multiple budgets, and assume that only on-path choices are observed. Our model allows for very general preferences, nesting the quasi-hyperbolic model, as well as many other models of intertemporal choice. The behavioral models that we consider involve naive and sophisticated agents. A naive Alice will make week-one choices unawares that her week-two self may not comply with her planned choices in week one. A sophisticated Alice will engage in strategic behavior vis-a-vis her week-two self. She will choose consumption in week one, fully internalizing the optimal response that her preferences will implement in week two. A sophisticated Alice's choices are equilibrium outcomes of a game between the preferences she holds at different points in time.

We provide characterizations of the datasets that are rationalizable by means of the naive, and of the sophisticated, models. The key ideas, as in most studies of revealed preference theory, is to connect the theoretical behavioral models to the configurations in the data that allow us (as analysists) to draw conclusions about the direction of the agents' preferences. In other words, to identify the configurations of data that define the correct notion of revealed preference. The axioms then take the form of standard ``strong axiom of revealed preference,'' acyclicity, conditions. These are discussed in detail in Section~\ref{sec:generalmodel}, where we also provide examples that illustrate when the axioms are violated. Our results are stated in Section~\ref{sec:characterizations} and proven in Section~\ref{sec:proofs}.

\subsection*{Related Literature}
To deal with the problem of unobservability of dynamic inconsistency, behavioral economists usually look for what \cite{o1999doing} call ``smoking guns,'' such as a costly effort to constraint one's future actions. A review of the empirical literature may be found in \cite{frederick2002time}.

Separately from the pure models of dynamic inconsistency emphasized in our paper, there are good reasons to think that attitudes towards risk and uncertainty lie behind some forms of non-exponential discounting. These arguments are developed in \cite{halevy2008strotz} and \cite{chakraborty2020relation}.

The closest paper to our work is \cite{blowbrowningcrawford2021}, who provide a characterization of the datasets that are consistent with the model of a  sophisticated quasi-hyperbolic consumer. We argue here that their notion of rationalization does not coincide with our definition, and discuss some of the differences. We take their work as a motivation for ours, but there are of course important differences. Their tests are truly meant for survey data, and therefore they insist on a single observation for each consumer. Our datasets require multiple observations --- which is standard, on the other hand, in revealed preference theory (for example, \cite{afriat1967construction}, who does not explicitly deal with dynamic decision making). As applicable to intertemporal choice, our test are more suitable for laboratory experiments. We discuss these issues further in Section~\ref{sec:discussion}.

When turning to the model of abstract choice in Section~\ref{sec:generalmodel}, we follow the tradition in \cite{arrow1959rational} of pursuing the general empirical content of rational choice, aside from the specific structure of consumer choice in Euclidean consumption spaces and Walrasian budgets. The closest work to ours in this setting is \cite{ray2001game} and \cite{bossertsprumont}, who consider the testable implications of game-theoretic models of dynamic choice (in a sense, following up on \cite{mcgarvey1953theorem}, who studies binary voting trees). \cite{ray2001game} characterize the testable implications of subgame-perfect Nash equilibrium when all subgames of a given extensive form game are observed, which essentially means that ``planned" (off-path) choices are observed. \cite{bossertsprumont} characterize the testable implications (and lack thereof) for choice that is the outcome of an unobserved dynamic strategic interaction (see also the extension by \cite{REHBECK2014207}). In our model, in contrast, the extensive-form game in question is observable and fixed by the dynamic multiple-self interaction we have described in the introduction. What is not observed in our model are the agents' off-path choices.\footnote{The literature on revealed preference and game theoretic models is reviewed in \cite{chambers2016revealed}.} Finally, we should mention \cite{manzinimariotti2007} who characterize a model of sequential choice but restricted to the same choice problem. This means that they do not face the same equilibrium questions that are present in our model.

The work of \cite{echenique2020testable} characterizes the testable implications of several different models of intertemporal choice, including quasi-hyperbolic discounting with a concave utility function. Their work, however, is restricted to a single planned choice made in the first period of choice. Our paper instead considers on-path choices, when plans are not observable.

The issues with the first-order approach that we discuss in Section~\ref{sec:bbcqhd} are, to some extent, predicted by \cite{pelegyaari1973}, who point out the difficulties with preserving convexity of preferences when solving for equilibrium by backward induction (see, however, \cite{goldman1980}).

\section{A special model: quasi-hyperbolic discounting and consumption choice.}\label{sec:bbcqhd}

We first turn our attention to dynamically inconsistent consumer choice, and in particular to quasi-hyperbolic discounting: the most popular model used to capture dynamically inconsistent choices in applied economic theory.

We focus on a three-period model because it is the simplest case in which the assumption of hyperbolic discounting has any bite. The three-period model is really a model of two-period choice, because period three is purely ``residual;'' which we then generalize in Section~\ref{sec:generalmodel}. We should, however, emphasize that the results we discuss here and in Section~\ref{sec:focsrat} hold for any finite number of periods. 

A consumer is choosing quantities of a single good in periods $t=1,2,3$. She has a wealth $m$, and faces prices $p_t$ for consumption in period $t$. These prices may be interpreted as encoding interest rates. In the paper, we normalize any price vector so that $p_3=1$. Given prices and wealth, a consumption stream $(x_1,x_2,x_3)$ is affordable if $p\cdot x= \sum_{t=1}^3p_tx_t\leq m$.

The standard exponential-discounting model assumes that preferences over a consumption stream $x=(x_1,x_2,x_3)\in\Re^3_+$ are described by a pair $(u,\da)$, with $u:\Re_+\to\Re$, and $\da>0$. The consumer evaluates a consumption stream $x=(x_1,x_2,x_3)$ by  
\[ 
u(x_1)+\da u(x_2) + \da^2 u(x_3),
\] and chooses an optimal affordable consumption stream.
We shall soon impose additional assumptions on $u$ and $\da$.

When her preferences satisfy the exponential discounting model, we can either imagine the consumer choosing a consumption stream as a plan, or alternatively as choosing consumption over time. The assumption of exponential discounting means that the consumer is dynamically consistent. When choosing over time, she has no reason to deviate from her planned consumption.

Under quasi-hyperbolic discounting, the consumer's preferences are described by a tuple $(u,\beta,\da)$, with $u:\Re_+\to\Re$, and $\beta,\da>0$. The consumer evaluates a consumption path $x=(x_1,x_2,x_3)$ by
\[
u(x_1) + \beta[\da u(x_2) + \da^2 u(x_3)].
\]

A sophisticated quasi-hyperbolic consumer chooses consumption that results from an equilibrium between their period-1 preferences and their period-2 preferences. We phrase this as a game played between two agents.\footnote{The game-theoretic formulation is standard in the literature. See, for example, \cite{o1999doing} or \cite{harris2001dynamic}.} Agent~1 chooses consumption in period 1, $x_1$. Agent~2 chooses consumption in period~2, and thus in period~3 because consumption in period 3 is determined by the consumer's overall budget. So Agent~2 chooses $(x_2,x_3)$. 

The relevant equilibrium notion embodies a form of sequential rationality: it is a subgame-perfect Nash equilibrum. A subgame-perfect equilibrium can be described by backward induction: In period 2, given $x_1$, agent 2 maximizes 
\[u(x_1)+\beta\delta u(x_2),\]
subject to $x_2,x_3\geq 0$ and $p_2x_2+p_3x_3\leq m-p_1 x_1$. Let  $s(x_1)=(s_2(x_1),s_3(x_1))$ denote a solution to Agent 2's problem, as a consumption vector in periods 2 and 3, and as a function of the period-1 choice $x_1$.

Agent 1 then solves the problem of choosing period-1 consumption $x_1$ to maximize
\[
u(x_1)+ \beta \delta u(s_2(x_1)) + \beta \delta^2 u(s_3(x_1)),
\] subject to $x_1\geq 0$ and  $p_1x_1\leq m$. If $x^*_1$ is an optimal choice for Agent 1, we say that the pair $(x^*_1,s)$ is a \df{subgame-perfect Nash equilibrium} of the game induced by $(u,\beta,\da)$. In the sequel, we simply write equilibrium to refer to a subgame-perfect Nash equilibrium.

An \df{equilibrium outcome} of the game defined by $(u,\beta,\da)$, for fixed prices $p=(p_1,p_2,p_3)$ and budget $m$, is then a consumption stream $x=(x_1,x_2,x_3)$ for which there exists an equilibrium $(x^*_1,s)$ with $x^*_1=x_1$ and $(x_2,x_3)=s(x_1)$.

When $u$ is smooth ($C^2$), we observe that a sufficient condition for an interior equilibrium outcome is the equation:
\begin{equation}\label{eq:FOCsrat}
u'(x_t) = \la \frac{p_t}{\da^t}\prod_{i=1}^t\frac{1}{1-(1-\beta)\mu_i},   
\end{equation}
where $\la$ is a Lagrange multiplier. The numbers $\mu_t$ are marginal propensities to consume from wealth. 

A standard calculation, assuming interior solutions and using the Implicit Function Theorem (which requires $u''\neq 0$ at a solution), yields that 
\begin{equation}\label{eq:strongfocs}
    \mu_2 = \frac{\beta\da p^2_2\, u''(x_3)}{u''(x_2)+ \beta\da p^2_2\, u''(x_3)},
\end{equation}
while $\mu_3=1$ ($\mu_1$ plays no role in the equilibrium analysis as $\lambda$ is a free parameter).

Equation~\eqref{eq:FOCsrat} is due to Definition 1 of \cite{blowbrowningcrawford2021} and Equation \eqref{eq:strongfocs} is due to Lemma 1 of  \cite{blowbrowningcrawford2021} (see also \cite{harris2001dynamic}). We also derive them directly in Section \ref{sec:prep1}.

\subsection{Data}\label{sec:bbcdata}
A \df{dataset} is a pair $(x,p)$, where $x\in\Re^3_{++}$ and $p\in \Re^3_{++}$. The interpretation is that we observe a consumption stream $x$, chosen when the prices are $p$, and income (or budget) is $m=p\cdot x$. Importantly, $x$ is the observed, or realized, consumption choice. Not a plan. 

Note that we restrict attention to data with strictly positive consumption. We will discuss the viability of the ``first-order approach,'' so we assume away the obvious question of whether corner solutions could mean that interior first-order conditions are incorrect. The points we wish to make are orthogonal to  the existence of corner choices (in fact, our main results continue to hold if we assume data to be non-negative).

Finally, in this section we shall restrict attention to datasets with a single observation; see the discussion in Section~\ref{sec:discussion}. 

\subsection{Rationalizability and FOCs rationalizability}
We now introduce the relevant notions of rationalizability: what it means for a dataset to be consistent with a particular theory of consumer choice.

Let $\mathcal{U}$ be the set of all monotone increasing, $C^2$, and concave functions $u:\Re_+\to\Re$. And let $\mathcal{U}_{+}$ be the set of all $u\in \U$ that are strictly concave.

\begin{defn}\label{def:expdisct}
A dataset $(x,p)$ is \df{rationalizable by the exponential discounting model} if there exist $u\in \mathcal{U}$ and $\da\in (0,1]$ for which $x$ solves the problem of maximizing  $u(z_1)+\da u(z_2)+\da^2 u(z_3)$ over the budget set 
\[
\{z\in\Re^3_+: p_1 z_1+ p_2 z_2 + z_3\leq p\cdot x\}.
\] When this occurs, we say that $(u,\da)$ is a \df{rationalization} of the data by means of the exponential discounting model.
\end{defn}

We note that one may be interested in exponential discounting without imposing smoothness or concavity of utility, but we follow \cite{blowbrowningcrawford2021} in these assumptions because we want to address the results in their paper.

\begin{defn}
A dataset $(x,p)$ is \df{equilibrium rationalizable by the sophisticated quasi-hyperbolic model} if there exists $(u,\beta,\da)$, with $u\in\mathcal{U}_{+}$, $\beta\in(0,1)$, and $\da\in (0,1]$, for which $(x_t)_{t=1}^3$ is an equilibrium outcome of the game defined by $(u,\beta,\da)$ for prices $p$ and budget $m=p\cdot x$. When this occurs, we say that $(u,\da,\beta)$ is a rationalization of the data by means of the sophisticated quasi-hyperbolic discounting model.    
\end{defn}

Again, we follow \cite{blowbrowningcrawford2021} in imposing strict concavity of utility, $\da\leq 1$ and $\beta<1$.\footnote{The definition in \cite{blowbrowningcrawford2021} explicitly requires concavity, not strict concavity; but they assume that the consumption function is differentiable, which (as far as we know) relies on an application of the Implicit Function Theorem that rules out a zero of the second derivative of the instantaneous utility function. Blow et al.\ refer to \cite{harris2001dynamic}, who do assume strict concavity of utility. So do \cite{krusell2003consumption}, in a paper that studies a closely related model.} Indeed, Blow et al.\ emphasize that if one allows for $\da>1$ then their conditions amount to checking for GARP.


The definitions of rationalizability by exponential discounting, or by the sophisticated quasi-hyperbolic discounting models, require checking complicated optimization and equilibrium properties. The literature on revealed preference theory, following the seminar work of \cite{afriat1967construction}, often focuses on the data satisfying the first-order conditions in the model.

\begin{defn}\label{defn:focsrats} 
\leavevmode
\begin{enumerate}
\item[i)] A dataset $(x,p)$ is \df{FOCs rationalizable as exponential discounting} if there is $\da>0$ and $u\in \U$ so that 
\[
u'(x_t) = \frac{p_t}{p_{t+1}}\da u'(x_{t+1}).
\] Such a pair $(u,\da)$ is a FOCs rationalization of exponential discounting. 

   \item[ii)] A dataset $(x,p)$ is \textbf{FOCs rationalizable} by the sophisticated quasi-hyperbolic model if there exists $(u,\beta,\da,(\mu_t)_{t=1}^3)$ such that $u\in \mathcal{U}_{+}$, $\la>0$, $\beta\in(0,1)$, $\da\in (0,1]$, $\mu_t\in (0,1)$ for $t=1,2$, and $\mu_3=1$; so that Equation (\ref{eq:FOCsrat}) is satisfied.
We say that the triple $(u,\beta,\da,(\mu_t)_{t=1}^3)$ is a \df{FOCs rationalization by the sophisticated quasi-hyperbolic model}.

\item[iii)] A dataset $(x,p)$ is \textbf{strong FOCs rationalizable} by the sophisticated quasi-hyperbolic model if there exists a FOCs rationalization $(u,\beta,\da,(\mu_t)_{t=1}^3)$ that satisfies Equation (\ref{eq:strongfocs}).\smallskip

\end{enumerate}
\end{defn}

Definition \ref{defn:focsrats}(\textit{i}) simply imposes the well-known first-order condition (or Euler equation) for maximizing a discounted utility function. FOCs rationalizability for exponential discounting was developed, and applied, by \cite{browning1989nonparametric}.

Definition~\ref{defn:focsrats}(\textit{ii}) of FOCs rationalizability for sophisticated quasi-hyperbolic discounting was introduced by \cite{blowbrowningcrawford2021}. FOCs rationalizability does not connect the numbers $\mu_t$ with the rationalizing instantaneous utility function $u$, so we present the stronger definition~\ref{defn:focsrats}(\textit{iii}) which imposes such a connection by way of Equation~\eqref{eq:strongfocs}.

\cite{afriat1967construction} shows that the first-order approach is valid for the general model of utility maximization in consumer theory. A similar statement is true for exponential discounting. It is presented here without proof.

\begin{proposition}\label{prop:expdisct}
A dataset is FOCs rationalizable by the exponential discounting model if and only if it is rationalizable by the exponential discounting model. Moreover, any FOCs rationalization is also a rationalization as the exponential discounting model.
\end{proposition}

\subsection{FOCs rationalizability for the quasi-hyperbolic model}\label{sec:focsrat}

Let FOC (respectively SFOC and EQ) be the set of all datasets that are FOCs (respectively strong FOCs and equilibrium) rationalizable. Let D be the set of datasets $(x,p)$ that satisfy $x_t\neq x_s$ for all $t\neq s$, and I be the set of datasets with $p_1>p_2>p_3$ and $p_1/p_2> p_2/p_3$.

\begin{theorem}\label{thm:example}\leavevmode
\begin{enumerate}
    \item $\textrm{EQ} \subseteq \textrm{SFOC}\subsetneq \textrm{FOC}$, 
    \item $\textrm{D}\cap \textrm{FOC}\subsetneq \textrm{SFOC}$, 
    \item and $\textrm{I}\subseteq \textrm{FOC}$.
\end{enumerate}
\end{theorem} 

A proof of Theorem~\ref{thm:example} is in Section~\ref{sec:proofs}, as are the proofs of all results in the paper.

To unpack the theorem, we discuss the different claims it contains. First it states the obvious logical relation among the different notions or rationalization:  $\textrm{EQ} \subseteq \textrm{SFOC}\subseteq \textrm{FOC}$; but in contrast with the message of Proposition~\ref{prop:expdisct} for exponential discounting, there is a gap between the notion of equilibrium and FOCs rationalization (EQ $\neq$ FOC), with the latter being strictly more permissive. The gap already shows up in comparing FOCs and strong FOCs rationalizations, so that SFOC is a strict subset of FOC.

The second statement addresses the, perhaps, most glaring issue with the notion of FOCs rationalization: The disconnect between $u$ and $\mu_t$s. One could try to amend the first-order approach by insisting on strong FOCs, but Theorem~\ref{thm:example} says that, as long as consumption in different time periods is distinct, it is always possible to line up the $\mu_t$ numbers with the intended rationalization. So strong FOCs seems to be too permissive as well.\footnote{See Theorem~\ref{thm:manyobs} in Appendix A for further evidence on this claim. When we allow for multiple observations, we exhibit a dataset that is in SFOC but not in EQ.}

The third statement provides additional evidence about the weakness of FOCs rationalization. No matter what the values of consumption are, as long as a dataset satisfies the assumption on prices in I, then it is FOCs rationalizable. It is worth mentioning that such prices are compatible with data that refute the exponential discounting model.\footnote{For example, consumption $x=(3,2,4)$ at prices $p=(8,2,1)$ violate the SAR-EDU of \cite{echenique2020testable} by means of the balanced sequence $(x_1,x_2),(x_3,x_2)$.}

Our next result speaks to the use of FOCs rationalizability to recover, or estimate, the parameters of a quasi-hyperbolic utility function.

\begin{theorem}\label{thm:delta} \leavevmode
\begin{enumerate}
    \item Let $\delta^*\in (0, 1)$ and $\beta^*\in (0, 1)$. There is an equilibrium rationalizable dataset $(x,p)$ such that: a) $\delta=\delta^*$ and $\beta=\beta^*$ for any equilibrium rationalization $(u, \beta, \delta)$ of $(x,p)$, and b) there is also a FOCs rationalization $(u, \beta', \delta')$ of $(x,p)$ with $\delta'=1$.
\item
There exists $(x,p)\in \textrm{SFOC}$ with a (strong FOCs) rationalization $(u,\beta,\da)$ so that $x$ is not an equilibrium outcome of $(u,\beta,\da)$.
\end{enumerate}
\end{theorem}

Estimating $\beta$ and $\da$ is an important empirical exercise. Discount factors matters critically for welfare comparisons, and is a key input in policy decisions. But the proof of Proposition 1 in \cite{blowbrowningcrawford2021} shows that whenever a dataset is FOCs rationalizable, it is without loss of generality to assume $\delta=1$. Theorem~\ref{thm:delta} claims that the reason is the permissiveness of FOCs. The assumption that $\delta<1$, for example, has additional empirical content when we focus on equilibrium rationalizability rather than FOCs.

Finally, the second statement in Theorem~\ref{thm:delta} speaks to the possibility of using a FOCs, or strong FOCs, rationalization in order to recover utility parameters. The theorem says that a rationalization may not have an equilibrium outcome that coincides with the data, which would mean that the rationalizing parameters could not generate the observed data. This problem is a particular challenge for the analysis in Section~3.4 of \cite{blowbrowningcrawford2021}, in which the authors set out to recover preferences based on a FOCs rationalization. The recovered preferences may not explain the data according to the sophisticated quasi-hyperbolic model.

\section{A general model of two-period choice}\label{sec:generalmodel}

We now turn to an analysis of a general model of two-period choice, again with two periods being the bare minimum needed to discuss dynamically inconsistent choices. The case of more than two periods is treated in Section~\ref{sec:morethan2}, where it becomes clear that the two-period analysis already captures the empirical content of sophisticated dynamically inconsistent choice.

We assume given a set $X=X_1\times X_2$ of possible alternatives, or choices, that can be made. First, Agent 1 chooses $x_1\in X_1$. Second, after observing 1's choice of $x_1$, Agent 2 chooses $x_2\in X_2$. The outcome is the pair $\mathbf{x}=(x_1,x_2)\in X$; and agents have potentially different preferences $\succsim_1$ and $\succsim_2$ over $X$. We allow for very general preferences, nesting the quasi-hyperbolic model discussed until now, as well as many other models of intertemporal choice.

When agents 1 and 2 have  different preferences, the outcome may be dynamically inconsistent in the following sense: Agent 1 chooses $x_1\in X_1$ as part of a ``plan'' $(x_1,x'_2)\in X$. Agent 2, however, may not comply with the intended plan of 1, and choose $x_2\in X_2$ according to their preference $\succsim_2$. 

For example, to accommodate the setting in Section~\ref{sec:bbcqhd}, we can have $X_1=\Re_+$ and $X_2=\Re^2_+$. In this case, Agent 1 chooses period-1 consumption $x_1$ and has preferences over $X$ represented by $u(x_1)+\beta\da u(y_2)+\beta\da^2u(y_3)$; while  Agent 2 chooses consumption in periods 2 and 3, $x_2=(y_2,y_3)\in\Re^2_+$, with preferences over $X$ that are represented by $u(y_2)+\beta\da u(y_3)$. The source of dynamic inconsistency is that 1's marginal rate of substitution for consumption in periods 2 and 3 is $\da u'(y_3)/u'(y_2)$, while it is  $\beta \da u'(y_3)/u'(y_2)$ for Agent 2.

More generally, any setting with choices over many periods can be accommodated by considering choices in the first period and the rest. In stationary environments with a recursive formulation, one could study the full dynamic problem by means of two-period decisions.

Our focus is on observable choices from a family of feasible sets of choices: \df{budgets} $B\subseteq X$. Given a budget $B$, we denote by $B_1=\{x_1\in X_1: (x_1,x_2)\in B \text{ for some } x_2\in X_2\}$ its projection onto $X_1$; and by $B_2(x_1)=\{x_2\in X_2: (x_1,x_2)\in B \}$ its section at $x_1\in X_1$.  Given a budget $B$, Agent 1 may choose $x_1\in B_1$ as part of a plan $(x_1,x'_2)\in B$. Then Agent 2 chooses $x_2\in B_2(x_1)$. For example, in the model of Section~\ref{sec:bbcqhd}, a budget is defined by prices $p$ and income $m$, and takes the form $B(p,m)=\{(x_1,y_2,y_3)\in X:p_1x_1+p_2y_2+p_3y_3\leq m \}$.

\subsection{Behavioral assumptions}

Now we may envision different behavioral assumptions. The first, naive, model has Agent 1 choosing a plan  $(x_1,x'_2)\in B$, unawares that Agent 2 is not committed to following the plan, and may choose $x_2\neq x'_2$. The {\em observed} outcome is then $(x_1,x_2)$. As analysts, we never get to see 1's intended $x'_2$.

The second, ``Nash,'' model assumes that Agent 1 chooses $x_1\in B_1$ whilst correctly predicting that Agent 2 will choose $x_2\in B_2(x_1)$. This model corresponds to a Nash equilibrium between the players, if they were to move simultaneously. As analysts, we again observe the realized choices $(x_1,x_2)$.

Finally, the fully sophisticated model envisions the dynamic game between Agent 1 and 2 that we have emphasized in Section~\ref{sec:bbcqhd}. Agent 1 fully understands that Agent 2 is choosing $x_2$ according to her preferences, and that her choices are constrained to $B_2(x_1)$. This model is Stackeleberg to the Nash model's Cournot. Key to Agent 1's behavior is her understanding of all the choices that 2 will make in response to different $x_1$. As analysts, however, we only observe the on-path choices $(x_1,x_2)$. We must infer that $x_2$ is optimal for Agent 2 in $B_2(x_1)$, and that $x_1$ is optimal for Agent 1 among all the pairs $(x'_1,x'_2)$ that she could obtain by selecting $x'_1\in B_1$, and predicting an optimal choice $x'_2\in B_2(x'_1)$ by Agent 2.

\subsection{Data}

Key to our analysis is the assumption of choice from multiple budgets. So we depart from the assumption of a single budget in Section~\ref{sec:bbcqhd}. A \df{dataset} is a collection $\mathcal{D}=\{\mathbf{x}^k, B^k\}_{k\in K}$ where $B^k\subseteq X$ and $\mathbf{x}^k\in B^k$. Each dataset is comprised of a finite collection of observations $(\exb^k,B^k)$, where $B^k$ is a budget and $\exb^k=(x^k_1,x^k_2)\in B^k$. The interpretation is that, when faced with budget $B^k$, Agent 1 chose $x^k_1\in B^k_1$ and Agent 2 chose $x^k_2\in B^k_2(x^k_1)$.

\subsubsection{Notational conventions}

Preference relations, which are denoted as $\succsim$, over $X$ are assumed to be complete and transitive (that is, weak orders). The set of maximal elements according to a preference $\succsim$ is denoted by \[
\max(B, \succsim)=\{\mathbf{x}\in B| \mathbf{x}\succsim \textbf{y}\text{ for any }\textbf{y}\in B\}.
\]

Finally, we use $B^k(x_1)=\{\mathbf{x}\in X|\exists x_2\in X_2\text{ s.t. }\mathbf{x}=(x_1, x_2)\in B^k\}$ and
$B^k(\cdot, x_2)=\{\mathbf{x}\in X|\exists x_1\in X_1\text{ s.t. }\mathbf{x}=(x_1, x_2)\in B^k\}$.\footnote{Note that $B^k(x_1)$ is a subset of $X$ while $B^k_2(x_1)$ is a subset of $X_2$.}

\subsubsection{Notions of rationalization}
Given our notion of data, and the different behavioral models one may use to explain it, there are still some questions regarding the exact nature of a rationalization. For example, if one allows for preferences that are indifferent among all outcomes, then in principle there are no testable implications from any model. The problem of discipline is discussed in detail in \cite{chambers2016revealed}, who propose two different paradigms: weak and strong rationalization. Here we are essentially adopting the strong viewpoint, which means that we insist that the predicted theoretical outcomes exactly match the observed choices. (The alternative, weak, notion would only require that observations are among the theoretically predicted outcomes, but that preferences satisfy some additional properties that rule out trivial rationalizations.)

We first consider the naive model outlined above, whereby the agent who moves first does not know, or realize, that the second agent has different preferences.

A data set $\mathcal{D}$ is \df{naively rationalizable} if there exist preference relations $\succsim_1$ and $\succsim_2$ on $X$ such that for each observation $k\in K$,\begin{equation}\label{EqN1}
\max\big(B^k, \succsim_1\big)\in B^k(x^k_1)\end{equation}
and
\begin{equation}\label{EqN2}
\max\big(B^k(x^k_1), \succsim_2\big)=\mathbf{x}^k.
\end{equation}

The meaning of $\max\big(B^k, \succsim_1\big) \in B^k(x^k_1)$ is that there is a uniquely optimal choice for $\succsim_1$ in $B^k$, and it involves choosing $x^k_1\in X_1$. It also involves some unobserved $z^k_2\in X_2$, so $\max\big(B^k, \succsim_1\big) = (x^k_1,z^k_2)$, but the rationalization is free to construct this counterfactual ``planned'' choice. The second choice, made by Agent 2, is then the best element $\mathbf{x}^{k}$ of $B^k(x^k_1)$.


A data set $\mathcal{D}$ is \df{naively Nash rationalizable} if there exist preference relations $\succsim_1$ and $\succsim_2$ on $X$ such that for each observation $k\in K$,\begin{equation}\label{EqNN1}
\mathbf{x}^k\in \max\big(B^k(\cdot, x^k_2), \succsim_1\big)\end{equation}
and
\begin{equation}\label{EqNN2}
\max\big(B^k(x^k_1), \succsim_2\big)=\mathbf{x}^k.
\end{equation}

Going back to the discussion of the different possible behavioral assumptions one could use, this definition says that there are preferences for which the observed outcomes constitute a Nash equilibrium between Agent 1 and 2.

One may desire a stricter discipline on the rationalization, and impose that the observed choice by Agent 1 is their unique optimal action. This leads to the next notion of rationalization.

A data set $\mathcal{D}$ is \df{strictly naively Nash rationalizable} if there exist preference relations $\succsim_1$ and $\succsim_2$ such that for each observation $k\in K$,\begin{equation}\label{EqSNN1}
\max\big(B^k(\cdot, x^k_2), \succsim_1\big)=\mathbf{x}^k\end{equation}
and
\begin{equation}\label{EqSNN2}
\max\big(B^k(x^k_1), \succsim_2\big)=\mathbf{x}^k.
\end{equation}





A data set $\mathcal{D}$ is \df{sophisticated rationalizable} if there exist preferences $\succsim_1$ and $\succsim_2$ on $X$ such that for each observation $k\in K$,
\begin{equation}\label{Eq5}
\max\Big(\bigcup_{x_1\in B^k_{1}}\max\big(B^k(x_1), \succsim_2\big), \succsim_1\Big)=\mathbf{x}^k.\end{equation}

It should be clear that the definition corresponds to our third behavioral model, in which Agent 2 will choose, for each possible $x_1\in B^k_1$, a maximal element of $B^k(x_1)$. Given that, the Agent 1 chooses $x^k_1$ that maximizes her preference and the second agents chooses $x^k_2$ given $x^k_1$. In the setting of Section~\ref{sec:bbcqhd}, sophisticated rationalizability corresponds to sophisticated quasi-hyperbolic rationalizability.

\subsection{Axioms}\label{sec:axioms}

A property analogous to the strong-axiom of revealed preference (SARP) is obviously going to be important in characterizing rationality. The first such axiom is easy to formulate and interpret, once we understand that $\mathbf{x}^{k_l}\in B^{k_{l+1}}(x^{k_{l+1}}_1)$ means that $\exb^{k_{l+1}}$ is revealed preferred to $\exb^{k_l}$:

\begin{namedaxiom}[N-SARP]
    There is no sequence $k_1, \ldots, k_L$ of $K$ such that for each $l\le L$, $\mathbf{x}^{k_l}\neq \mathbf{x}^{k_{l+1}}$ and 
\begin{equation}\label{Eq3}
\mathbf{x}^{k_l}\in B^{k_{l+1}}(x^{k_{l+1}}_1).\end{equation}
\end{namedaxiom}

Figure~\ref{fig:nsarp} exhibits a violation of N-SARP. In the figure, $X_1=\{x_1,x'_1\}$ and $X_2=\{x_2,x'_2,x''_2 \}$. There are two observations, shown on the left and the right in the figure. In the figure on the left, we have that $B_1$ is either $\{x_1\}$ or $\{x'_2\}$. The feasible choices for Agent 2 are illustrated in the game tree. So, $B^1=X\setminus \{(x'_1, x''_2)\}$ and $B^2=X\setminus \{(x_1, x''_2)\}$. Agents choices are depicted in blue, $\mathbf{x}^1=(x_1, x_2)$ and $\mathbf{x}^2=(x_1, x'_2)$.  The figure shows a violation of N-SARP because $\mathbf{x}^1 \in B^2(x^2_1)$ and $\mathbf{x}^2 \in B^1(x^1_1)$. 

\begin{figure}
\begin{tikzpicture}[x=0.5cm, y=0.7cm][domain=0:1, range=0:1, scale=3/4, thick]
\usetikzlibrary{intersections}

\coordinate[label=above:${P1}$] (wf) at (7.5, 10.2);
\coordinate[label=above:${P2}$] (wf) at (2.5, 7);
\coordinate[label=above:${P2}$] (wf) at (10.6, 7);

\filldraw[] (7.5, 10) circle (3pt);
\filldraw[] (10, 7) circle (3pt);
\filldraw[] (3, 7) circle (3pt);
\filldraw[] (0, 4) circle (3pt);
\filldraw[] (3, 4) circle (3pt);
\filldraw[] (5, 4) circle (3pt);
\filldraw[] (13, 4) circle (3pt);
\filldraw[] (8.5, 4) circle (3pt);

\coordinate[label=left:${x_1}$] (wf) at (5, 8.5);
\coordinate[label=left:${x'_1}$] (wf) at (10.2, 8.5);

\coordinate[label=left:${x_2}$] (wf) at (1.5, 5.5);
\coordinate[label=right:${x'_2}$] (wf) at (1.8, 5.5);
\coordinate[label=right:${x''_2}$] (wf) at (4, 5.5);
\coordinate[label=left:${x_2}$] (wf) at (9.3, 5.5);
\coordinate[label=right:${x'_2}$] (wf) at (11.5, 5.5);

\draw[blue][ultra thick](7.5, 10)--(3, 7);
\draw[->](7.5, 10)--(10, 7);
\draw[blue][ultra thick](3, 7)--(0, 4);
\draw[->](3, 7)--(3, 4);
\draw[->](3, 7)--(5, 4);
\draw[->](10, 7)--(13, 4);
\draw[->](10, 7)--(8.5, 4);


\coordinate[label=below:${x_1x_2}$] (wf) at (0, 4);

\coordinate[label=below:${x_1x'_2}$] (wf) at (3, 4);

\coordinate[label=below:${x_1x''_2}$] (wf) at (5.2, 4);

\coordinate[label=below:${x'_1x'_2}$] (wf) at (13, 4);


\coordinate[label=below:${x'_1x_2}$] (wf) at (8.5, 4);

\coordinate[label=above:${P1}$] (wf) at (25.5, 10.2);
\coordinate[label=above:${P2}$] (wf) at (20.9, 7);
\coordinate[label=above:${P2}$] (wf) at (28.6, 7);

\filldraw[] (25.5, 10) circle (3pt);
\filldraw[] (28, 7) circle (3pt);
\filldraw[] (21.5, 7) circle (3pt);
\filldraw[] (18, 4) circle (3pt);
\filldraw[] (22, 4) circle (3pt);
\filldraw[] (31, 4) circle (3pt);
\filldraw[] (28, 4) circle (3pt);
\filldraw[] (25, 4) circle (3pt);

\coordinate[label=left:${x_1}$] (wf) at (23.5, 8.5);
\coordinate[label=left:${x'_1}$] (wf) at (28.2, 8.5);
\coordinate[label=left:${x_2}$] (wf) at (19.8, 5.5);
\coordinate[label=right:${x'_2}$] (wf) at (21.8, 5.5);
\coordinate[label=left:${x_2}$] (wf) at (26.5, 5.5);
\coordinate[label=right:${x'_2}$] (wf) at (26.8, 5.5);
\coordinate[label=right:${x''_2}$] (wf) at (29.5, 5.5);

\draw[blue][ultra thick](25.5, 10)--(21.5, 7);
\draw[->](25.5, 10)--(28, 7);
\draw[->](21.5, 7)--(18, 4);
\draw[blue][ultra thick](21.5, 7)--(22, 4);
\draw[->](28, 7)--(31, 4);
\draw[->](28, 7)--(28, 4);
\draw[->](28, 7)--(25, 4);


\coordinate[label=below:${x_1x_2}$] (wf) at (18, 4);

\coordinate[label=below:${x_1x'_2}$] (wf) at (22, 4);

\coordinate[label=below:${x'_1x'_2}$] (wf) at (28, 4);
\coordinate[label=below:${x'_1x''_2}$] (wf) at (31, 4);


\coordinate[label=below:${x'_1x_2}$] (wf) at (25, 4);

\end{tikzpicture}
\caption{A Violation of N-SARP}\label{fig:nsarp}
\end{figure}
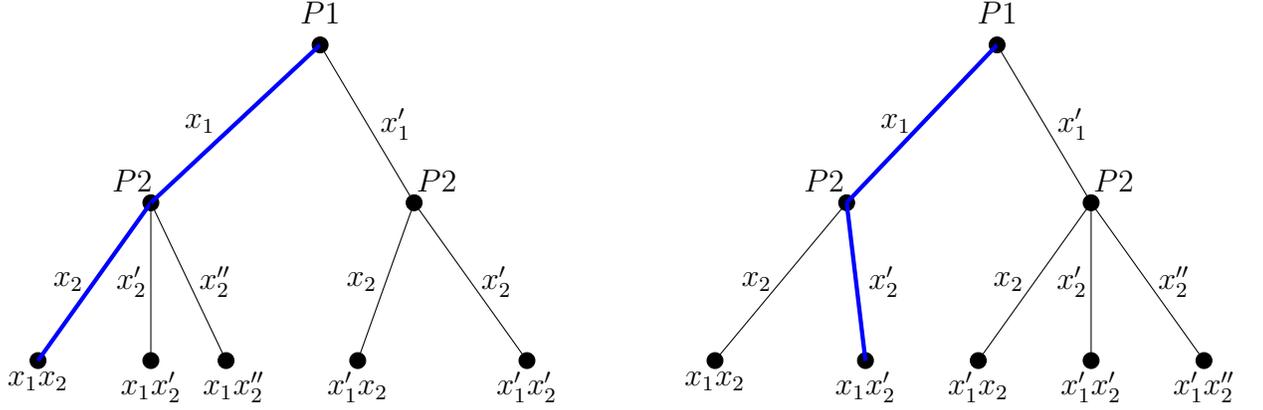

In the current model, the sections of two budget sets that correspond two different period-1 choices do not intersect, so the acyclicity for Agent 2's choices has no bite when Agent 1's choice are always distinct.

\bigskip
\noindent\textbf{Remark 1.} When $x^k_1\neq x^s_1$ for any $k, s\in K$, N-SARP is vacuous. 
\bigskip

A similar notion of acyclicity is meaningful for an Agent 1 who takes Agent 2's choice as given. If we denote by $B^k(\cdot,x^k_2)$ the section of $B^k$ at $x^k_2$, then we obtain our next version of SARP:

\bigskip

\noindent\textbf{NN-SARP.} There is no sequence $k_1, \ldots, k_L$ of $K$ such that for each $l\le L$, $\mathbf{x}^{k_l}\neq \mathbf{x}^{k_{l+1}}$ and 
\begin{equation}\label{EqNN}
\mathbf{x}^{k_l}\in B^{k_{l+1}}(\cdot, x^{k_{l+1}}_2).\end{equation}
\medskip

We now consider an axiom that tackles the non-observability of plans in the naive model. Imagine a dataset where $B^k(x^k_1)\subseteq B^{k'}\setminus B^{k'}(x^{k'}_1)$. Then we must infer that, no matter what Agent 1 planned to choose in period 2 from $B^{k'}(x^{k'}_1)$, it must be preferred by Agent 1 to whatever outcome she planned to choose in period 2 from $B^k(x^k_1)$. This means, then, that we could not have $B^{k'}(x^{k'}_1)\subseteq B^{k}\setminus B^{k}(x^{k}_1)$, as that would lead to the opposite conclusion. In consequence, $B^k(x^k_1)\cup B^{k'}(x^{k'}_1) \subseteq B^{k}\setminus B^{k}(x^{k}_1)\cup B^{k'}\setminus B^{k'}(x^{k'}_1)$ leads to a refutation of the naive model. 

More generally, the property we need is:

\begin{namedaxiom}[Condition 1] For any subset $S$ of $K$,
\begin{equation}\label{Eq4}
\bigcup_{k\in S}B^{k}(x^{k}_1)\not\subseteq \bigcup_{k\in S}B^{k}\setminus B^k(x^{k}_1).\end{equation}
\end{namedaxiom}

The choices in Figure~\ref{fig:nsarp} satisfy Condition 1, but 
Figure~\ref{fig:cond1} exhibits a violation of this condition. In the figure, we have the same $X$ and budgets as in the previous example, but now  $\mathbf{x}^1=(x'_1, x_2)$ and $\mathbf{x}^2=(x_1, x_2)$.  No matter what Agent 1 plans to choose in period 2 following $x_1$, this option was available to her in $B^1$ when she picked $x'_1$ from $B^1$. So it is clear that such data is incompatible with the naive model, and  Condition 1 is violated.

\begin{figure}
\begin{tikzpicture}[x=0.5cm, y=0.7cm][domain=0:1, range=0:1, scale=3/4, thick]
\usetikzlibrary{intersections}

\coordinate[label=above:${P1}$] (wf) at (7.5, 10.2);
\coordinate[label=above:${P2}$] (wf) at (2.5, 7);
\coordinate[label=above:${P2}$] (wf) at (10.6, 7);

\filldraw[] (7.5, 10) circle (3pt);
\filldraw[] (10, 7) circle (3pt);
\filldraw[] (3, 7) circle (3pt);
\filldraw[] (0, 4) circle (3pt);
\filldraw[] (3, 4) circle (3pt);
\filldraw[] (5, 4) circle (3pt);
\filldraw[] (13, 4) circle (3pt);
\filldraw[] (8.5, 4) circle (3pt);

\coordinate[label=left:${x_1}$] (wf) at (5, 8.5);
\coordinate[label=left:${x'_1}$] (wf) at (10.2, 8.5);

\coordinate[label=left:${x_2}$] (wf) at (1.5, 5.5);
\coordinate[label=right:${x'_2}$] (wf) at (1.8, 5.5);
\coordinate[label=right:${x''_2}$] (wf) at (4, 5.5);
\coordinate[label=left:${x_2}$] (wf) at (9.3, 5.5);
\coordinate[label=right:${x'_2}$] (wf) at (11.5, 5.5);

\draw(7.5, 10)--(3, 7);
\draw[blue][ultra thick](7.5, 10)--(10, 7);
\draw(3, 7)--(0, 4);
\draw[->](3, 7)--(3, 4);
\draw[->](3, 7)--(5, 4);
\draw[->](10, 7)--(13, 4);
\draw[ultra thick][blue](10, 7)--(8.5, 4);


\coordinate[label=below:${x_1x_2}$] (wf) at (0, 4);

\coordinate[label=below:${x_1x'_2}$] (wf) at (3, 4);

\coordinate[label=below:${x_1x''_2}$] (wf) at (5.2, 4);

\coordinate[label=below:${x'_1x'_2}$] (wf) at (13, 4);


\coordinate[label=below:${x'_1x_2}$] (wf) at (8.5, 4);

\coordinate[label=above:${P1}$] (wf) at (25.5, 10.2);
\coordinate[label=above:${P2}$] (wf) at (20.9, 7);
\coordinate[label=above:${P2}$] (wf) at (28.6, 7);

\filldraw[] (25.5, 10) circle (3pt);
\filldraw[] (28, 7) circle (3pt);
\filldraw[] (21.5, 7) circle (3pt);
\filldraw[] (18, 4) circle (3pt);
\filldraw[] (22, 4) circle (3pt);
\filldraw[] (31, 4) circle (3pt);
\filldraw[] (28, 4) circle (3pt);
\filldraw[] (25, 4) circle (3pt);

\coordinate[label=left:${x_1}$] (wf) at (23.5, 8.5);
\coordinate[label=left:${x'_1}$] (wf) at (28.2, 8.5);
\coordinate[label=left:${x_2}$] (wf) at (19.8, 5.5);
\coordinate[label=right:${x'_2}$] (wf) at (21.8, 5.5);
\coordinate[label=left:${x_2}$] (wf) at (26.5, 5.5);
\coordinate[label=right:${x'_2}$] (wf) at (26.8, 5.5);
\coordinate[label=right:${x''_2}$] (wf) at (29.5, 5.5);

\draw[blue][ultra thick] (25.5, 10)--(21.5, 7);
\draw(25.5, 10)--(28, 7);
\draw[blue][ultra thick](21.5, 7)--(18, 4);
\draw (21.5, 7)--(22, 4);
\draw[->](28, 7)--(31, 4);
\draw[->](28, 7)--(28, 4);
\draw(28, 7)--(25, 4);


\coordinate[label=below:${x_1x_2}$] (wf) at (18, 4);

\coordinate[label=below:${x_1x'_2}$] (wf) at (22, 4);

\coordinate[label=below:${x'_1x'_2}$] (wf) at (28, 4);
\coordinate[label=below:${x'_1x''_2}$] (wf) at (31, 4);


\coordinate[label=below:${x'_1x_2}$] (wf) at (25, 4);

\end{tikzpicture}
\caption{A Violation of Condition 1 (but not Condition 2)}\label{fig:cond1}
\end{figure}
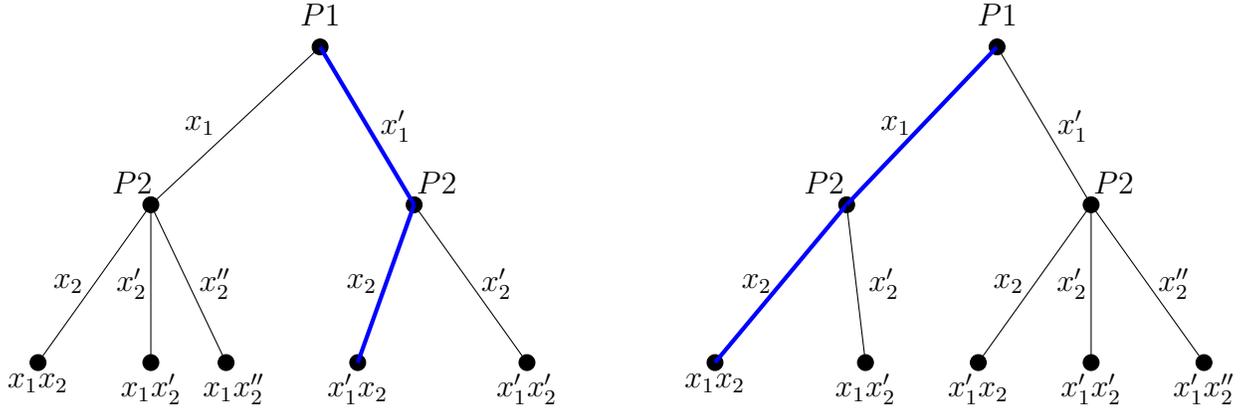

Next, to get a handle on Agent 1's revealed preferences, consider a situation where  
$\mathbf{x}^{k}\in B^{k'}(x^{k}_1)\subseteq B^{k}(x^{k}_1).$ This means that 2's choice in $B^k$ remains feasible in budget $B^{k'}$ after 1 chooses $x^k_1$, and all of 2's feasible choices in budget $B^{k'}$ after 1 chooses $x^k_1$ were feasible in feasible $B^k$. Since 2 picked $x^k_2$ as an optimal response to $x^k_1$ in budget $B^k$, she would still choose it as an optimal response to $x^k_1$ in budget  $B^{k'}$. We can then infer, by the usual logic of one object being chosen when another was feasible, that Agent 1 reveals preferred $x^{k'}_1$ to $x^k_1$. The acyclicity of this revealed preference relation is expressed as

\begin{namedaxiom}[Condition 2]
    There is no sequence $k_1, \ldots, k_L$ of $K$ such that for each $l\le L$, $\mathbf{x}^{k_l}\neq \mathbf{x}^{k_{l+1}}$ and  \begin{equation}\label{Eq7} \mathbf{x}^{k_l}\in B^{k_{l+1}}(x^{k_l}_1)\subseteq B^{k_l}(x^{k_l}_1).\end{equation}
\end{namedaxiom}

Condition 2 is satisfied in the examples we presented in Figures~\ref{fig:nsarp} and~\ref{fig:cond1}. Consider, however, the situation in Figure~\ref{fig:cond3}. The choice sets and budgets are as before, but  $\mathbf{x}^1=(x_1, x_2)$ and $\mathbf{x}^2=(x'_1, x'_2)$. We have $B^2(x_1)\subseteq B^1(x_1)$ and $\exb^1\in B^2(x_1)$, which means that Agent 2 knows that by choosing $x_1\in B^2_1$ she will ensure $\exb^1$. Choosing $x'_1\in B_1$ that is different from $x_1$ reveals a preference for this choice. However, we may make the opposite inference if we reason from the fact that $B^1(x'_1)\subseteq B^2(x'_1)$. This data will, then, contradict the sophisticated model.

It is worth mentioning that the data in Figure~\ref{fig:cond3} satisfy Condition 1, N-SARP, and NN-SARP.

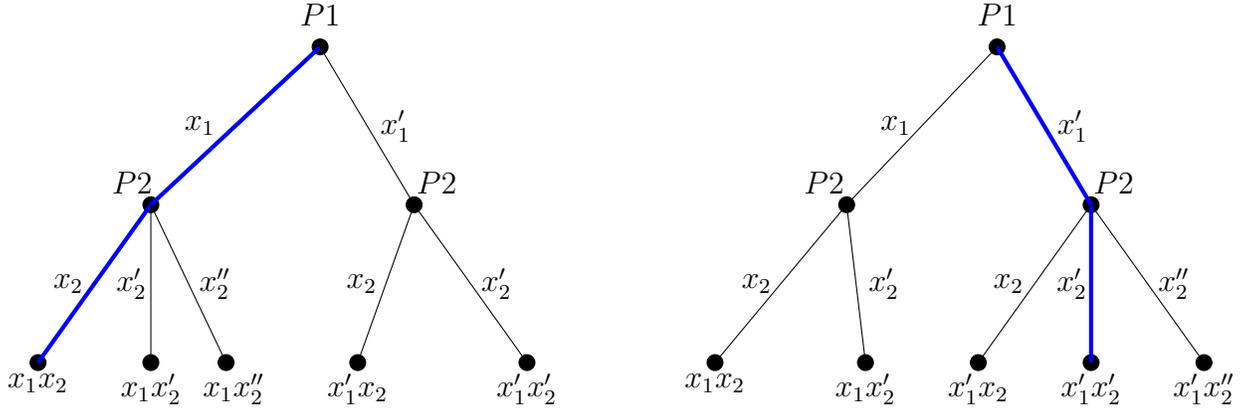
\begin{figure}
\begin{tikzpicture}[x=0.5cm, y=0.7cm][domain=0:1, range=0:1, scale=3/4, thick]
\usetikzlibrary{intersections}

\coordinate[label=above:${P1}$] (wf) at (7.5, 10.2);
\coordinate[label=above:${P2}$] (wf) at (2.5, 7);
\coordinate[label=above:${P2}$] (wf) at (10.6, 7);

\filldraw[] (7.5, 10) circle (3pt);
\filldraw[] (10, 7) circle (3pt);
\filldraw[] (3, 7) circle (3pt);
\filldraw[] (0, 4) circle (3pt);
\filldraw[] (3, 4) circle (3pt);
\filldraw[] (5, 4) circle (3pt);
\filldraw[] (13, 4) circle (3pt);
\filldraw[] (8.5, 4) circle (3pt);

\coordinate[label=left:${x_1}$] (wf) at (5, 8.5);
\coordinate[label=left:${x'_1}$] (wf) at (10.2, 8.5);

\coordinate[label=left:${x_2}$] (wf) at (1.5, 5.5);
\coordinate[label=right:${x'_2}$] (wf) at (1.8, 5.5);
\coordinate[label=right:${x''_2}$] (wf) at (4, 5.5);
\coordinate[label=left:${x_2}$] (wf) at (9.3, 5.5);
\coordinate[label=right:${x'_2}$] (wf) at (11.5, 5.5);

\draw[blue][ultra thick](7.5, 10)--(3, 7);
\draw[->](7.5, 10)--(10, 7);
\draw[blue][ultra thick](3, 7)--(0, 4);
\draw[->](3, 7)--(3, 4);
\draw[->](3, 7)--(5, 4);
\draw[->](10, 7)--(13, 4);
\draw[->](10, 7)--(8.5, 4);


\coordinate[label=below:${x_1x_2}$] (wf) at (0, 4);

\coordinate[label=below:${x_1x'_2}$] (wf) at (3, 4);

\coordinate[label=below:${x_1x''_2}$] (wf) at (5.2, 4);

\coordinate[label=below:${x'_1x'_2}$] (wf) at (13, 4);


\coordinate[label=below:${x'_1x_2}$] (wf) at (8.5, 4);

\coordinate[label=above:${P1}$] (wf) at (25.5, 10.2);
\coordinate[label=above:${P2}$] (wf) at (20.9, 7);
\coordinate[label=above:${P2}$] (wf) at (28.6, 7);

\filldraw[] (25.5, 10) circle (3pt);
\filldraw[] (28, 7) circle (3pt);
\filldraw[] (21.5, 7) circle (3pt);
\filldraw[] (18, 4) circle (3pt);
\filldraw[] (22, 4) circle (3pt);
\filldraw[] (31, 4) circle (3pt);
\filldraw[] (28, 4) circle (3pt);
\filldraw[] (25, 4) circle (3pt);

\coordinate[label=left:${x_1}$] (wf) at (23.5, 8.5);
\coordinate[label=left:${x'_1}$] (wf) at (28.2, 8.5);
\coordinate[label=left:${x_2}$] (wf) at (19.8, 5.5);
\coordinate[label=right:${x'_2}$] (wf) at (21.8, 5.5);
\coordinate[label=left:${x_2}$] (wf) at (26.5, 5.5);
\coordinate[label=right:${x'_2}$] (wf) at (26.8, 5.5);
\coordinate[label=right:${x''_2}$] (wf) at (29.5, 5.5);

\draw (25.5, 10)--(21.5, 7);
\draw[blue][ultra thick](25.5, 10)--(28, 7);
\draw[->](21.5, 7)--(18, 4);
\draw (21.5, 7)--(22, 4);
\draw[->](28, 7)--(31, 4);
\draw[blue][ultra thick] (28, 7)--(28, 4);
\draw(28, 7)--(25, 4);


\coordinate[label=below:${x_1x_2}$] (wf) at (18, 4);

\coordinate[label=below:${x_1x'_2}$] (wf) at (22, 4);

\coordinate[label=below:${x'_1x'_2}$] (wf) at (28, 4);
\coordinate[label=below:${x'_1x''_2}$] (wf) at (31, 4);


\coordinate[label=below:${x'_1x_2}$] (wf) at (25, 4);

\end{tikzpicture}

\caption{A Violation of Condition 2 (but not Condition 1)}\label{fig:cond3}
\end{figure}

\subsection{Characterizations}\label{sec:characterizations}

Our strongest result requires the assumption that no two observed first-period choices are the same. For data that satisfy the assumptions, we can precisely characterize the naive and the sophisticated model. It becomes clear, then, that the gap between these models is exactly the gap between Conditions 1 and 2.

\begin{theorem}\label{thm:sophdistinct} Let $\mathcal{D}$ be a dataset in which  $x^k_1\neq x^s_1$ for any two distinct $k, s\in K$. 
\begin{enumerate}
\item $\D$ is naively rationalizable iff it satisfies Condition 1.
    \item $\D$ is sophisticated rationalizable iff it satisfies Condition 2.
    \item $\D$ is strictly naively Nash rationalizable iff it satisfies NN-SARP.
\end{enumerate}
\end{theorem}

When the assumption in Theorem~\ref{thm:sophdistinct} is not satisfied, we can still characterize the naive and naively-Nash models.

\begin{thm}\label{thm:N} Let $\D$ be a dataset. 
\begin{enumerate}
    \item $\D$ is naively rationalizable iff it satisfies N-SARP and Condition 1.  
    \item $\D$ is naively Nash rationalizable iff it satisfies N-SARP. 
    \item  $\D$ is strictly naively Nash rationalizable iff it satisfies N-SARP and NN-SARP.  
\end{enumerate}
\end{thm}

\subsection{Example on NN-SARP}

We end with an example exhibiting a violation of NN-SARP. Again the choice and budget sets are as in the previous examples. Consider an example where $B^1=X\setminus \{(x'_1, x''_2)\}$ and $B^2=X\setminus \{(x_1, x''_2)\}$. 

Consider the choices in Figure~\ref{fig:NNsarp}, $\mathbf{x}^1=(x_1, x_2)$ and $\mathbf{x}^2=(x'_1, x_2)$. In this example, NN-SARP is violated (consequently Condition 2 is also violated) since $\mathbf{x}^1\in B^2(\cdot, x^2_2)$ and $\mathbf{x}^2\in B^1(\cdot, x^1_2)$, but Condition 1 and N-SARP are satisfied.





\begin{figure}
\begin{tikzpicture}[x=0.5cm, y=0.7cm][domain=0:1, range=0:1, scale=3/4, thick]
\usetikzlibrary{intersections}

\coordinate[label=above:${P1}$] (wf) at (7.5, 10.2);
\coordinate[label=above:${P2}$] (wf) at (2.5, 7);
\coordinate[label=above:${P2}$] (wf) at (10.6, 7);

\filldraw[] (7.5, 10) circle (3pt);
\filldraw[] (10, 7) circle (3pt);
\filldraw[] (3, 7) circle (3pt);
\filldraw[] (0, 4) circle (3pt);
\filldraw[] (3, 4) circle (3pt);
\filldraw[] (5, 4) circle (3pt);
\filldraw[] (13, 4) circle (3pt);
\filldraw[] (8.5, 4) circle (3pt);

\coordinate[label=left:${x_1}$] (wf) at (5, 8.5);
\coordinate[label=left:${x'_1}$] (wf) at (10.2, 8.5);

\coordinate[label=left:${x_2}$] (wf) at (1.5, 5.5);
\coordinate[label=right:${x'_2}$] (wf) at (1.8, 5.5);
\coordinate[label=right:${x''_2}$] (wf) at (4, 5.5);
\coordinate[label=left:${x_2}$] (wf) at (9.3, 5.5);
\coordinate[label=right:${x'_2}$] (wf) at (11.5, 5.5);

\draw[blue][ultra thick](7.5, 10)--(3, 7);
\draw[->](7.5, 10)--(10, 7);
\draw[blue][ultra thick](3, 7)--(0, 4);
\draw[->](3, 7)--(3, 4);
\draw[->](3, 7)--(5, 4);
\draw[->](10, 7)--(13, 4);
\draw[->](10, 7)--(8.5, 4);


\coordinate[label=below:${x_1x_2}$] (wf) at (0, 4);

\coordinate[label=below:${x_1x'_2}$] (wf) at (3, 4);

\coordinate[label=below:${x_1x''_2}$] (wf) at (5.2, 4);

\coordinate[label=below:${x'_1x'_2}$] (wf) at (13, 4);


\coordinate[label=below:${x'_1x_2}$] (wf) at (8.5, 4);

\coordinate[label=above:${P1}$] (wf) at (25.5, 10.2);
\coordinate[label=above:${P2}$] (wf) at (20.9, 7);
\coordinate[label=above:${P2}$] (wf) at (28.6, 7);

\filldraw[] (25.5, 10) circle (3pt);
\filldraw[] (28, 7) circle (3pt);
\filldraw[] (21.5, 7) circle (3pt);
\filldraw[] (18, 4) circle (3pt);
\filldraw[] (22, 4) circle (3pt);
\filldraw[] (31, 4) circle (3pt);
\filldraw[] (28, 4) circle (3pt);
\filldraw[] (25, 4) circle (3pt);

\coordinate[label=left:${x_1}$] (wf) at (23.5, 8.5);
\coordinate[label=left:${x'_1}$] (wf) at (28.2, 8.5);
\coordinate[label=left:${x_2}$] (wf) at (19.8, 5.5);
\coordinate[label=right:${x'_2}$] (wf) at (21.8, 5.5);
\coordinate[label=left:${x_2}$] (wf) at (26.5, 5.5);
\coordinate[label=right:${x'_2}$] (wf) at (26.8, 5.5);
\coordinate[label=right:${x''_2}$] (wf) at (29.5, 5.5);

\draw (25.5, 10)--(21.5, 7);
\draw[blue][ultra thick](25.5, 10)--(28, 7);
\draw[->](21.5, 7)--(18, 4);
\draw (21.5, 7)--(22, 4);
\draw[->](28, 7)--(31, 4);
\draw[->](28, 7)--(28, 4);
\draw[blue][ultra thick](28, 7)--(25, 4);


\coordinate[label=below:${x_1x_2}$] (wf) at (18, 4);

\coordinate[label=below:${x_1x'_2}$] (wf) at (22, 4);

\coordinate[label=below:${x'_1x'_2}$] (wf) at (28, 4);
\coordinate[label=below:${x'_1x''_2}$] (wf) at (31, 4);


\coordinate[label=below:${x'_1x_2}$] (wf) at (25, 4);

\end{tikzpicture}

\caption{A Violation of NN-SARP (Consequently, Condition 2)}\label{fig:NNsarp}
\end{figure}
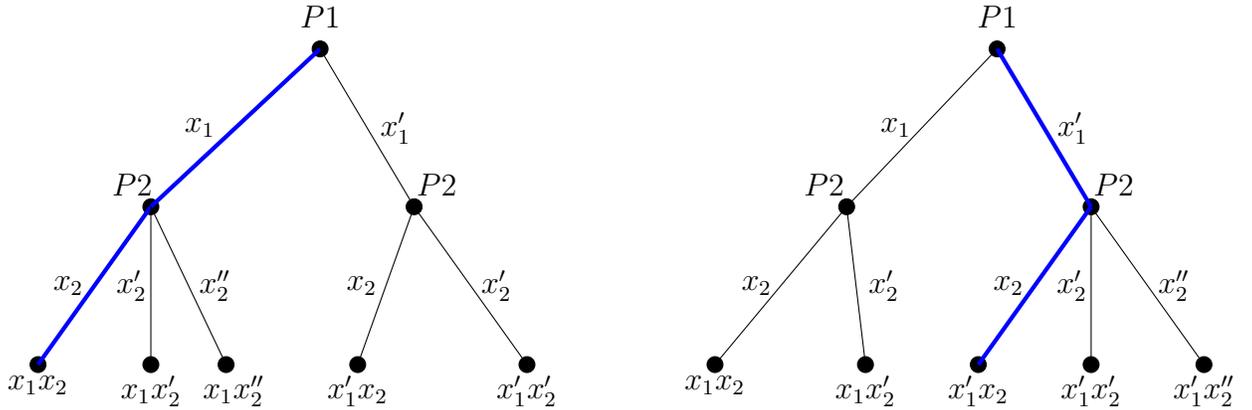

\section{Sophisticated and naive rationalizations with more than two periods}\label{sec:morethan2}

We now extend our analysis of sophisticated and naive rationalizations for choice problems with
more than two periods. Let $T\ge 2$ be the number of periods and $X=\prod^T_{t=1}
X_t$ be the set of all possible alternatives. First, Agent 1 chooses $x_1\in X_1$ in
period $1$. Then, for any $t\le T$, after observing choices $\mathbf{x}_{t-1}$, Agent
$t$ chooses $x_t\in X_t$ in period $t$. The outcome is the vector $\mathbf{x}=(x_1,x_2, \ldots, x_T)\in X$; and agents have potentially different preferences $\{\succsim_t\}_{t\in T}$ over $X$.

For example, to accommodate the quasi-hyperbolic discounting model, we can have $X_t=\Re_+$ for each $t<T$ and $X_T=\Re^2_+$. In this case, Agent $t<T$ chooses period-t consumption $x_t$ and has preferences $\succsim_t$ over $X$ represented by $u(x_t)+\sum^{T-t-1}_{s=1}\beta\da^{s} u(x_{t+s})+\beta \delta^{T-t} u(y_T)+\beta \delta^{T-t+1} u(y_{T+1})$, while agent $T$ chooses consumption in periods $T$ and $T+1$, $x_T=(y_T, y_{T+1})\in \Re^2_+$, with preferences over $X$ that are represented by $u(y_T)+\beta\delta u(y_{T+1})$. It is immediate that when $\beta<1$, preferences $\succsim_t$ are different.

\medskip
\noindent\textbf{Notations.} We write $\mathbf{x}_t=(x_1, \ldots, x_t)$ and
$\mathbf{x}_{-t}=(x_1, \ldots, x_{t-1}, x_{t+1}, \ldots, x_T)$ for a given
$\mathbf{x}=(x_1, \ldots, x_T)$ and $t\in T$. Given a budget set $B\subseteq X$, we
also write \[B(\mathbf{x}_{t-1})=\{\mathbf{x}\in X|\exists (x_{t}, \ldots, x_T)\text{ s.t. }\mathbf{x}\in B\},\] 
\[B_t(\mathbf{x}_{t-1})=\{x_t\in X_t|\exists (x_{t+1}, \ldots, x_T)\text{ s.t. }\mathbf{x}\in B\},\]
and 
\[B_{-t}(\mathbf{x}_{-t})=\{\mathbf{x}\in X|\exists x_{t}\in X_t\text{ s.t. }\mathbf{x}\in B\}.\]
For simplicity, we abuse notation and denote $B_{-t}(\mathbf{x}_{-t})$ as $B(\mathbf{x}_{-t})$.

A \df{dataset} is a collection $\mathcal{D}=\{\mathbf{x}^k, B^k\}_{k\in K}$ where
$B^k\subseteq X$ is a budget and $\mathbf{x}^k\in B^k$. Each dataset is comprised of a finite collection of observations $(\exb^k,B^k)$, where $B^k$ is a budget and $\exb^k=(x^k_t)_{t\in T}\in B^k$. The interpretation is that, when faced with budget $B^k$, Agent $t$ chose $x^k_t\in B^k_t(\mathbf{x}^k_{t-1})$.

A data set $\mathcal{D}$ is \df{naively rationalizable} if there exist preferences $\{\succsim_t\}_{t\in T}$ on $X$ such that for each observation $k\in K$,

\begin{equation}\label{Eq30}
\max\big(B^k(\mathbf{x}^k_{t-1}), \succsim_{t}\big)\in B^k(\mathbf{x}^k_{t})
\end{equation}
for each $t<T$ and 
\begin{equation}\label{Eq31}
\max(B^k(\mathbf{x}^k_{T-1}), \succsim_T)=\mathbf{x}^k.
\end{equation}
Note that Equation~\eqref{Eq30} says that Agent $t$'s optimal choice in
$B^k(\mathbf{x}^k_{t-1})$ coincide with $x^k_t$ in period $t$.

A data set $\mathcal{D}$ is \df{naively Nash rationalizable} if there exist preferences $\{\succsim_t\}_{t\in T}$ on $X$ such that for each observation $k\in K$,\begin{equation}\label{Eq36}
\mathbf{x}^k\in \max\big(B^k(\mathbf{x}^k_{-t}), \succsim_t\big)\end{equation}
for each $t<T$ and
\begin{equation}\label{Eq37}
\max\big(B^k(\mathbf{x}^k_{T-1}), \succsim_T\big)=\mathbf{x}^k.
\end{equation}

One may desire a stricter discipline on the rationalization, and impose that the observed choice by Agent $t$ is their unique optimal action. This leads to the next notion of rationalization.

A data set $\mathcal{D}$ is \df{strictly naively Nash rationalizable} if there exist preferences $\{\succsim_t\}_{t\in T}$ such that for each observation $k\in K$ and $t\le T$,\begin{equation}\label{Eq38}
\max\big(B^k(\mathbf{x}^k_{-t}), \succsim_t\big)=\mathbf{x}^k.\end{equation}

A data set $\mathcal{D}$ is \df{sophisticated rationalizable} if there exist preferences $\{\succsim_t\}_{t\in T}$ on $X$ such that for each observation $k\in K$,
\begin{equation}\label{Eq26}
\max\Big(M^k_1, \succsim_1\Big)=\mathbf{x}^k,\end{equation}
where $M^k_1, \ldots, M^k_T$ are defined recursively as follows: 
\begin{equation}\label{Eq27}
M^k_t(\mathbf{x}_{t-1})=\bigcup_{x_{t}\in B^k_{t}(\mathbf{x}_{t-1})}\max\big(M^k_{t+1}(\mathbf{x}_{t}), \succsim_{t+1}\big)
\end{equation}
for each $t<T$ and $M^k_T(\mathbf{x}_{T-1})=B^k(\mathbf{x}_{T-1})$.

\subsection{Axioms and Characterization.} To characterize the naive model in this environment, we provide appropriate modifications of N-SARP and Condition 1. For each $t\in T$, we define the following.

\begin{namedaxiom}[$t$-SARP]
    There is no sequence $k_1, \ldots, k_L$ of $K$ such that for each $l\le L$, $\mathbf{x}^{k_l}\neq \mathbf{x}^{k_{l+1}}$ and 
\begin{equation}\label{Eq32}
\mathbf{x}^{k_l}\in B^{k_{l+1}}(\mathbf{x}^{k_{l+1}}_{-t}).\end{equation}
\end{namedaxiom}

Note that N-SARP is equivalent to $2-$SARP and NN-SARP is equivalent to $1$-SARP when $T=2$.
\medskip

\begin{namedaxiom}[Condition 5] For any subset $S$ of $K$ and $t\le T$,
\begin{equation}\label{Eq33}
\bigcup_{k\in S}B^{k}(\mathbf{x}^k_{t})\not\subseteq \bigcup_{k\in S}B^{k}(\mathbf{x}^k_{t-1})\setminus B^k(\mathbf{x}^k_t).\end{equation}
\end{namedaxiom}

The behavioral implications of N-SARP and Condition 1 remain the same. Given our definitions of $B^k(x_1)$, Condition 2 is well-defined and has the same behavioral implications in this general environment as we show below.\medskip

Again, our strongest result requires the assumption that no two observed first-period choices are the same. For data that satisfy the assumptions, we can precisely characterize the naive and the sophisticated model.

\begin{theorem}\label{thm:sopgen} Let $\mathcal{D}$ be a dataset in which  $x^k_1\neq x^s_1$ for any two distinct $k, s\in K$. 
\begin{enumerate}
\item $\D$ is sophisticated rationalizable iff it satisfies Condition 2.
\item $\D$ is naively rationalizable iff it satisfies Condition 5.
\item $\D$ is strictly naively Nash rationalizable iff it satisfies $1$-SARP.
\end{enumerate}
\end{theorem}

\begin{remark} The sufficiency proof of Statement~(1) of Theorem~\ref{thm:sopgen}
  essentially reduces to the two-period problem treated in the body of the
  paper. When we construct a rationalizing preference, we rely heavily on the
  two-period result from Section~\ref{sec:characterizations}, and we argue that we
  may set the preferences of the period $t$ agent to equal those of the second period
  agent, for all $t>2$. In a sense, the only tension that has real empirical
  consequences is between the first and subsequent agents.
  \end{remark}

When the assumption in Theorem~\ref{thm:sopgen} is not satisfied, we can still characterize the naive and naively-Nash models.

\begin{theorem}\label{thm:naigen} Let $\mathcal{D}$ be a dataset.
\begin{enumerate}
\item $\D$ is naively rationalizable iff it satisfies $T$-SARP and Condition 5.
    \item $\D$ is naively Nash rationalizable iff it satisfies $T$-SARP.
    \item $\D$ is strictly naively Nash rationalizable iff it satisfies $t$-SARP for each $t\in T$.
\end{enumerate}
\end{theorem}

\section{Discussion and conclusion}\label{sec:discussion}

Our paper spans two different approaches to revealed preference theory. The first uses the structure of consumption space, in which assumptions regarding convexity (and possibly monotonicity) are meaningful and can be important. Smoothness assumptions are useful as well. To the extent that the theory is meant to be used with survey data, it is important to obtain testable implications even when we have a single observation of intertemporal consumption.

The second, more abstract setting, seeks to isolate the problem of testing for rationality from the consumption environment. Think of Arrow's (or Sen's) formulation of the choice problem, in contrast to the theory of Samuelson and Hicks. The simplification afforded by the abstract setting is that there is no need to ensure that rationalizing preferences possess any special structure, as long as some discipline avoids a trivial rationalization.  In the abstract setting, it is crucial to have multiple observations in order to obtain any testable implications.

We proceed to discuss these aspects of our paper further.

Our characterization results rely on datasets that involve multiple observations. This is a key difference from \cite{blowbrowningcrawford2021} (and the earlier work of \cite{browning1989nonparametric}), who are primarily interested in survey data of consumers making choices over time. In survey data, one would typically only observe each consumer making a unique choice in every period. Our results are, in contrast, applicable to experimental data. It is common in experiments to have participants make more than one choice, and then pay out a choice selected at random.

The multiplicity of observations in our model raises two issues. The first is that, perhaps, having access to many observations gives FOCs rationalizability more power to reject the property of being quasi-hyperbolic rationalizable by the sophisticated model. To this end, we present an example in  Appendix~\ref{sec:focsappendix} that shows how, with an arbitrarily large dataset, there is still a gap between the FOCs notion of rationalizability, and equilibrium rationalizability by means of the sophisticated quasi-hyperbolic model. Indeed the gap is, in that case, between strong FOCs rationalizability and equilibrium rationalizability.

The second issue regards the motivation.  If datasets with multiple observations make most sense for experimental data, then one could also design experiments in which plans are observable by having experiment participants first formulate a plan, and then coming back to the lab to revise their choices. Such experiments have been run, but are notoriously difficult to implement, and raise questions about the exact inferences that one can make from the agents' choices. So we still believe that there is value in having a characterization that will make sense for standard experimental designs implemented over time.

Turning to the abstract choice model, it allows us to get a handle on the sophisticated problem, which seems (to us) intractable for the consumption environment. In the abstract choice model, there is no structure on the rationalizing preferences, because there is really no structure in the choice environment. This lack of discipline on rationalizing preference does not inherently render the problem easier, but it allows us to identify the relevant revealed preference inference that can be made from observed intertemporal choices. In consequence, we obtain quite naturally necessary conditions that the data must satisfy to be consistent with the sophisticated model. Such necessary conditions turn out, for the most part, to be sufficient. Further necessary conditions are discussed and explored in Appendix~\ref{sec:appB}.

\section{Proofs}\label{sec:proofs}

As a preparation for the proof of Theorem~\ref{thm:example}, we derive convenient expressions for the model's first- and second-order conditions. 

\subsection{First-Order Conditions}\label{sec:prep1}  Agent 2's objective function is strictly concave, so the interior solutions to 2's optimization problem are characterized by the first-order condition
\begin{equation}\label{eq:foc2}
\frac{u'(x_2)}{u'(x_3)} = p_2\, \beta\,\delta.
\end{equation}

By strict concavity, the optimum is unique, and we may write the choice of $x_3$ as a function of $x_2$ by $g(x_2)$, where $g(x_2) = (u')^{-1} (A\, u'(x_2))$, and $A=\frac{1}{\beta\da p_2}$. We may also write period-1 consumption as \[
x_1 = f(x_2) =\frac{m-p_2\,x_2 - g(x_2)}{p_1}
\] (recall that $p_3=1$). Note that the function $g$ is smooth and strictly monotone increasing, while $f$ is smooth and strictly decreasing.

In Section~\ref{sec:bbcqhd}, we formulated the game between the two agents using optimal strategies $s_2$ and $s_3$ for Agent~2; but note that $s_2$ and $f$ are inverses of each other, and that $s_3$ is determined from $s_2$ by the budget constraint.\footnote{Indeed, the budget constraint is $p_2 s_2(x_1) + s_3(x_1) = m-p_1 x_1$ (when $p_3=1$). So $p_2 s'_2(x_1) + s'_3(x_1) = -p_1.$ Agent 1's FOC is: $u'(x_1) + \beta\da u'(x_2)s'_2(x_1) + \beta\da^2 u'(x_3)s'_3(x_1)=0$, which becomes $u'(x_1) = \beta\da[ -u'(x_2)s'_2 +\da u'(x_3)(p_1+p_2 s'_2)]$. So $\frac{u'(x_1)}{u'(x_2)} =\da \left[\frac{p_1}{p_2}+(1-\beta) s'_2\right]$ since $u'(x_3)=u'(x_2)/{\beta\delta p_2}$. The formulation in \cite{blowbrowningcrawford2021} uses a consumption function that depends on current period wealth. Let the consumption function be $c_2(A_2)$ and note that
$s_2(x_1) = c_2(m-p_1x_2),$ so $s'_2 = -p_1 c'_2 = - \frac{p_1}{p_2}\mu_2$;
where we have used that $\mu_t=p_t c'_t(A_t)$. Putting these together we get $\frac{u'(x_1)}{u'(x_2)} =\da\frac{p_1}{p_2} \left[1-(1-\beta) \mu_2\right]$. 
This is the equation used in Blow et.\ al.}

It is convenient to represent Agent 1's problem as choosing $x_2$ to maximize \[
u(f(x_2)) + \beta\delta\, u(x_2) + \beta\delta^2\, u(g(x_2)), 
\] subject to the relevant non-negativity constraints. The first-order condition for an interior solution is then
\begin{equation}\label{eq:focint}
u'(f(x_2))f'(x_2) + \beta\delta\, u'(x_2) + \beta\delta^2\, u'(g(x_2))g'(x_2)=0.
\end{equation} 
By definition of $f$ and $g$, we then obtain
\begin{equation}\label{eq:foc1}
\frac{u'(x_1)}{u'(x_2)}=\frac{\delta\, p_1}{p_2}\frac{g'(x_2)+\beta\, p_2}{g'(x_2)+p_2}.\end{equation}

\begin{lem}\label{lem:fg-sfocs}
    Let $(x,p)$ be a dataset and $(u,\beta,\da)\in \U_+\times (0,1)\times (0,1]$ be such that~\eqref{eq:foc2} and~\eqref{eq:foc1} are satisfied. Let $\mu_1\in (0,1)$ be arbitrary, \[
    \mu_2=\frac{p_2}{\frac{u''(x_2)}{p_2\,\beta\,\delta u''(x_3)}+p_2},
    \] and $\mu_3=1$, then Equations~\eqref{eq:FOCsrat} and~\eqref{eq:strongfocs} are satisfied. In other words, $(u,\beta,\da)$ is a strong FOCs rationalization of $(x,p)$.
\end{lem}
\begin{proof}
Define $\lambda=\frac{\delta\, u'(x_1)(1-(1-\beta)\mu_1)}{p_1}$. We then obtain Equation \eqref{eq:FOCsrat} for $t=1$.

From Equation \eqref{eq:foc1}, we obtain 
\[u'(x_2)=u'(x_1)\, \frac{p_2}{\delta\, p_1}\frac{g'(x_2)+p_2}{g'(x_2)+\beta\,p_2}=\lambda\, \frac{p_2}{\delta^2}\,\frac{1}{1-(1-\beta)\mu_1}\frac{g'(x_2)+p_2}{g'(x_2)+\beta\,p_2}\]
\[=\lambda\, \frac{p_2}{\delta^2}\,\prod^2_{i=1}\frac{1}{1-(1-\beta)\mu_i},\]
where we have used that 
\[
    \mu_2=\frac{p_2}{\frac{u''(x_2)}{p_2\,\beta\,\delta u''(x_3)}+p_2}= \frac{p_2}{g'(x_2)+p_2}\in (0, 1),
    \] as $g'(x_2)=\frac{A\,u''(x_2)}{u''(x_3)}$ and $A=\frac{1}{\beta\delta p_2}$.

Equation \eqref{eq:foc2}, implies that 
\[u'(x_3)=\frac{u'(x_2)}{p_2\,\beta \delta}=\lambda\, \frac{1}{\delta^3}\,\prod^2_{i=1}\frac{1}{1-(1-\beta)\mu_i}\frac{1}{\beta}=\lambda\, \frac{p_3}{\delta^3}\,\prod^3_{i=1}\frac{1}{1-(1-\beta)\mu_i},\]
where $\mu_3=1$ and $p_3=1$. In other words, we have derived Equation \eqref{eq:FOCsrat} for each $t\le 3$. 

Finally, we obtain Equation~\eqref{eq:strongfocs} from
\[\mu_2=\frac{p_2}{\frac{u''(x_2)}{p_2\,\beta\,\delta u''(x_3)}+p_2}=\frac{\beta\,\delta\, p^2_2\,u''(x_3)}{u''(x_2)+\beta\,\delta\,p^2_2\, u''(x_3)}.\]
\end{proof}

Indeed, strong FOCs rationalizability is equivalent to Equations~\eqref{eq:foc2} and ~\eqref{eq:foc1}. Note that the derivations hold for any $\mu_1\in (0, 1)$ as we can freely choose $\lambda$. Hence, the FOC rationalizability is equivalent to 
\[\frac{u'(x_1)}{u'(x_2)}=\delta\,\frac{p_1}{p_2}\,(1-(1-\beta)\,\mu_2)\text{ and }\frac{u'(x_3)}{u'(x_2)}=\beta\,\delta\,p_2.\]

\subsection{Deriving the Second-Order Conditions} Since Agent 1's objective function is  
\[u\big(f(x_2)\big)+\beta\,\delta\, u(x_2)+\beta\,\delta^2\,u\big(g(x_2)\big),\]
the SOC is 
\[u''(x_1)\,(f'(x_2))^2+\beta\,\delta\,u''(x_2)+\beta\,\delta^2\, u''(x_3)\,(g'(x_2))^2+g''(x_2)\big(\beta\,\delta^2\, u'(x_3)-\frac{u'(x_1)}{p_1}\big)\le 0.\]
Using Equations~\eqref{eq:foc2} and ~\eqref{eq:foc1}, we can further simplify and obtain
\begin{equation}\label{eq:SOCn}
u''(x_1)\,(f'(x_2))^2+\beta\,\delta\,u''(x_2)+\beta\,\delta^2\, u''(x_3)\,(g'(x_2))^2+g''(x_2)\,u'(x_2)\frac{\delta(1-\beta)}{g'(x_2)+p_2}\le 0.\end{equation}

When there is no present-bias, $\beta=1$, the above term is strictly negative since $u$ is strictly concave. Similarly, when $A=1$ (which occurs for some of our examples), we obtain $g(x_2)=x_2$ and again the second-order condition is satisfied. However, with present bias, $\beta<1$, or when $g''>0$, the term above can be strictly positive and the second-order condition violated.

\subsection{Proof of Theorem~\ref{thm:example}}\label{sec:pfthmexample} 

It is obvious that  $\textrm{EQ}\subseteq \textrm{SFOC}\subseteq \textrm{FOC}$. We proceed to show that the other statements in the theorem.

\emph{Part 1: There exists a dataset in FOC that is not in SFOC.} Consider a dataset with $x_1=x_2=x_3$, $p_2=3$, and $p_1=4$. This is FOCs rationalizable by the sophisticated quasi-hyperbolic model by setting $\la=\da=1$, $\beta=1/3$, $\mu_1=1/2$, $\mu_3=1$ and $\mu_2=3/8$. It is then easy to verify the definition of FOCs rationalizability.  We obtain that \[ 
    \la\, \frac{p_t}{\da^t}\prod_{i=1}^t\frac{1}{1-(1-\beta)\mu_i} = p_t\prod_{i=1}^t\frac{1}{1-(1-\beta)\mu_i}=6\] for $t=1,2,3$. Now setting $u'(x_t)=6$, and letting $u$ be any utility function in $\mathcal{U}_{+}$ that has derivative $=6$ at  the point $x_1$, renders the data FOCs rationalizable as desired.

The data is, however, not strong FOC rationalizable. Suppose it were, and let $(u,\beta,\da)$ be such a rationalization. Then~\eqref{eq:foc2} and $x_2=x_3$ imply that 
\[1=\frac{u'(x_2)}{u'(x_3)}=\beta\delta p_2.\] So $\beta\da = 1/3$. Consequently, $g(x_2) = (u')^{-1} (\frac{1}{\beta\da p_2} u'(x_2)) = x_2$, as $\beta\da p_2=1$. So $g'(x_2)=1$. On the other hand,~\eqref{eq:foc1} and $x_1=x_2$ imply that 
\[1=\frac{\delta\, p_1}{p_2}\frac{g'(x_2)+\beta\, p_2}{g'(x_2)+p_2}=\delta\,\frac{4}{3}\frac{1+3\beta}{1+3}=\frac{\delta+3\beta\delta}{3}=\frac{\delta+1}{3}\] and hence that $\da=2$. A contradiction.

\smallskip
\noindent\emph{Part 2. Any dataset in $FOC$ that has $x_t\neq x_s$ for all $t\neq s$ is strong FOCs rationalizable.} 

To prove this, we claim that, if $(\hat u,\beta,\da,(\mu_t)_{t=1}^3)$ is a FOCs rationalization, then we may find a strong FOCs rationalization $(u,\beta,\da,(\mu_t)_{t=1}^3)$ for which $u'(x_t)=\hat u'(x_t)$ for all $t$. To this end, let $a_t=\hat u'(x_t)>0$, and choose $b_t<0$ so that $\mu_2 = \frac{A^{-1}\, b_3\, p_2}{b_2+ A^{-1}\, b_3\, p_2}$ holds. Note that $(u,\beta,\da,(\mu_t)_{t=1}^3)$ will be a FOCs rationalization if $u'(x_t)=a_t$. Moreover, Equation~\eqref{eq:strongfocs} will be satisfied if  $u''(x_t)=b_t$.  

Consider the function $h_t(x) = a_t+b_t(x-x_t)$. Note that $h_t$ is monotone decreasing and that $a_t = h_t(x_t)<h_s(x_s)=a_s$ when $x_s<x_t$, as $\hat u$ is strictly concave. Given that $x_s\neq x_t$ for $t\neq s$ we may find disjoint neighborhoods $N_t$ of each $x_t$ so that $h_t$ is smaller on $N_t$ than $h_s$ on $N_s$ when $x_s<x_t$, and greater on $N_t$ than $h_s$ on $N_s$ when $x_s>x_t$. Define a function $h:\Re_+\to\Re$ by letting $h$ equal $h_t$ on $N_t$, $h(0)>\sup\{h_t(x):x\in N_t, 1\leq t\leq 3 \}$, and by linear interpolation on $\Re_+\setminus (\{0\}\cup (\cup_t N_t))$. Then $h$ is monotone decreasing, $h(x_t)=\hat u'(x_t)$, and $h'_t=b_t$ for all $t$. Letting $u(x)=\int_{0}^x h(z)\diff z$, we have $u'(x_t)=h(x_t)=a_t$ and $u''(x_t)=h'(x_t)=b_t$. 

\smallskip
\noindent\emph{Part 3.} $I\subseteq FOC$ directly follows from Proposition 1 of \cite{blowbrowningcrawford2021}.
 
\subsection{Proof of Theorem~\ref{thm:delta}}

\emph{Part 1:} To prove the first statement in the Theorem, fix $\da^*\in (0,1)$ and $\beta^*\in (0,1)$ and consider a dataset with $x_1=x_2=x_3$ and
\[p_2=\frac{1}{\beta^*\delta^*}\text{ and }
p_1 = \frac{1+\beta^*\delta^*}{(\beta^*\delta^*)^2(1+\delta^*)}.
\]

Let $(u,\beta,\da)$ be an equilibrium rationalization. We claim that $\da=\delta^*$ and $\beta=\beta^*$. 

Since $x_2=x_2$, by Equation \eqref{eq:foc2}, we have $1=\beta\da p_2$; i.e., $\beta\delta=\beta^*\delta^*$. Moreover, we also obtain $A=1$, which means that $g'(x_2)=1$. Since $x_1=x_2$, by Equation \eqref{eq:foc1}, we have 
\[1=\delta\,\frac{p_1}{p_2}\frac{g'(x_2)+\beta p_2}{g'(x_2)+p_2}
=
\delta\,\frac{1+\beta^*\delta^*}{\beta^*\delta^*(1+\delta^*)}\frac{1+\frac{\beta}{\beta^*\delta^*}}{1+\frac{1}{\beta^*\delta^*}}=\frac{\delta+1}{1+\delta^*}.\]
 Hence, $\da=\delta^*$ and $\beta=\beta^*$.

Observe that the data we have proposed are strong FOCs rationalizable, as we may
choose any strictly concave utility that has the requisite value of $u'(x_t)$ (given
$x_1=x_2=x_3$). In fact, it is equilibrium rationalizable because $\beta\delta p_2=1$ means that $g(x)=x$, which in turn implies that Agent 1's objective function, $u(f(x_2))+\beta\delta u(x_2)+\beta\delta^2 u(g(x_2))$, is strictly concave. So the dataset is certainly FOCs rationaliable.

On the other hand, the proof of Proposition 1 in \cite{blowbrowningcrawford2021}
shows that whenever a dataset is FOCs rationalizable then it is without loss of
generality to set $\da=1$. It is in fact easy to show that the data is FOCs rationalizable with $(u, \beta', \delta')$ with $\delta'=1$ and $\beta'=\beta^*\delta^*$ by setting $\mu_2=\frac{1-\beta^*(\delta^*)^2}{1-(\beta^*)^2(\delta^*)^2}$.
 
\emph{Part 2:} Finally, we prove the second statement in Theorem~\ref{thm:delta}.

Consider a dataset with $x_1=0.04, x_2=0.05$,  $x_3=0.4698$, and prices $p_2=2$ and $p_1=3.0969$ (consequently, $m=0.694$). We claim that $(u,\beta,\da)$, with $\beta=\delta=0.8$, and $u(x)=x-\frac{x^3}{3}$ when $x\in (0, 1)$, is a strong FOCs rationalization by the sophisticated quasi-hyperbolic model. 

Indeed, note that  $g(x)=\sqrt{1-A\,(1-x^2)}$ where $A=0.78125$. By direct calculation, we obtain \[g'(x_2)=\frac{A\, x_2}{\sqrt{1-A\,(1-x^2_2)}}=0.083147\text{ and }g''(x_2)=\frac{A\,(1-A)}{(1-A\,(1-x^2_2))^\frac{3}{2}}=1.648.\]
To verify strong FOCs rationalizability, note that Equation~\eqref{eq:foc2} is satisfied since
\[\frac{u'(x_2)}{u'(x_3)}=\frac{1-x^2_2}{1-x^3_2}=\frac{1-0.05^2}{1-0.4698^2}=\frac{1}{0.78125}=\beta\,\delta\, p_2.\]
Moreover, Equation~\eqref{eq:foc1} is satisfied since
\[\frac{u'(x_1)}{u'(x_2)}=\frac{1-0.04^2}{1-0.05^2}=\frac{0.8\times 3.1}{2} \left( \frac{0.083147+1.6}{0.083147+2}\right)=\frac{\delta\, p_1}{p_2} \left( \frac{g'(x_2)+\beta\, p_2}{g'(x_2)+p_2}\right).\] 

To check the equilibrium rationalizability, let us now consider the second order condition for Agent 1. However, we have 
\begin{align*}
  g''(x_2)\,u'(x_2)\frac{\delta(1-\beta)}{g'(x_2)+p_2}
  & =0.1263  > 0.1035 \\
  & =|u''(x_1)\,(f'(x_2))^2+\beta\,\delta\,u''(x_2)+\beta\,\delta^2\, u''(x_3)\,(g'(x_2))^2|,
\end{align*}
which means Equation~\eqref{eq:SOCn} is violated. Hence, the observed consumption bundle $(x_1, x_2, x_3)$ is a local minimizer for Agent 1's optimization problem. 

\subsection{Proof of Theorem \ref{thm:sophdistinct}.} Suppose $x^k_1\neq x^s_1$ for any $k, s\in K$ with $k\neq s$. Under this assumption, N-SARP becomes vacuous. Hence, the first and third parts of this theorem directly follow from Theorem~\ref{thm:N}. We now shall prove the second part of Theorem \ref{thm:sophdistinct}.

\medskip
\noindent\textbf{Necessity.} To prove the necessity of Condition 2, we define the following revealed preference relation.

\medskip
\noindent\textbf{Revealed Preference 1.} $\mathbf{x}^k \mathrel R_1 \mathbf{x}^s$ if $\mathbf{x}^k \neq \mathbf{x}^s$ and
\[\mathbf{x}^s\in B^k(x^s_1)\subseteq B^s(x^s_1).\]

To show the acyclicity of $R_1$, we will prove that $\mathbf{x}^k \mathrel R_1 \mathbf{x}^s$ implies $\mathbf{x}^k \succ_1 \mathbf{x}^s$. Take any $k, s$ such that $\mathbf{x}^k \mathrel R_1 \mathbf{x}^s$. Since 
\[\max\Big(\bigcup_{x_1\in B^s_{1}}\max\big(B^s(x_1), \succsim_{2}\big), \succsim_1\Big)=\mathbf{x}^s,\]
we have
\[\max\Big(\max\big(B^s(x^s_1), \succsim_{2}\big), \succsim_1\Big)=\mathbf{x}^s;\]
which implies
\[\mathbf{x}^s\in \max\big(B^s(x^s_1), \succsim_{2}\big).\]
Moreover, since 
\[\max\Big(\big\{\max\big(B^k(x_1), \succsim_{2}\big)\big\}_{x_1\in B^k_{1}}, \succsim_1\Big)=\mathbf{x}^k\text{ and }x^s_1\in B^k_1,\]
we have
\[\mathbf{x}^k\succ_1 \max\Big(\max\big(B^k(x^s_1), \succsim_{2}\big), \succsim_1\Big)\text{ or }\max\Big(\max\big(B^k(x^s_1), \succsim_{2}\big), \succsim_1\Big)=\mathbf{x}^k.\]
Since $\mathbf{x}^s\in B^k(x^s_1)$ and $B^k(x^s_1)\subseteq B^s(x^s_1)$,
\[\mathbf{x}^s\in \max\big(B^s(x^s_1), \succsim_{2}\big)\Rightarrow \mathbf{x}^s\in \max\big(B^k(x^s_1), \succsim_{2}\big).\]
Therefore, since $\mathbf{x}^s\in \max\big(B^k(x^s_1), \succsim_{2}\big)$, then we have $\mathbf{x}^k\succ_1\mathbf{x}^s$, the desired result.

\medskip
\noindent\textbf{Sufficiency.} We now show the sufficiency of Condition 2. By Condition 2, $R_1$ is acyclic. Let $\succsim_1$ be a preference relation that extends $R_1$ such that $\mathbf{x}^k\succ_1\mathbf{x}^s$ if $\mathbf{x}^k \mathrel R_1 \mathbf{x}^s$. Moreover, let $\mathbf{x}^k\succ_1\mathbf{x}$ for any $\mathbf{x}\in X\setminus \bigcup_{s\in K}\{\mathbf{x}^s\}$.


Let $\succsim_2$ be a preference relation such that 
\[\mathbf{x}^k\succ_{2} \mathbf{x}\text{ for any }\mathbf{x}\in B^k(x^k_1)\setminus\{\mathbf{x}^k\}\]
and 
\[\mathbf{x}\succ_{2}\mathbf{x}^k\text{ for any }\mathbf{x}\in \big(\{x^k_1\}\times X_2\big)\setminus B^k(x^k_1).\]

Observe that  $\succsim_2$ is well-defined. Note that there is no $\mathbf{x}$ such that $\mathbf{x}^k\succ_2 \mathbf{x}$ and $\mathbf{x}\succ_2 \mathbf{x}^2$ because the former requires that $\mathbf{x}$ be an element of $B^k(x^k_1)$, while the latter requires $\mathbf{x}$ to be not an element of $B^k(x^k_1)$. Moreover, longer cycles of $\succ_2$ cannot occur because $x^k_1\neq x^s_1$ for any $k, s\in K$ with $k\neq s$. Hence there exists a relation $\succsim_2$ that satisfies the stated properties.

Hence, 
\[\max\big(B^k(x^k_1), \succsim_{2}\big)=\mathbf{x}^k.\]


We shall show that
\[\max\Big(\bigcup_{x_1\in B^k_{1}}\max\big(B^k(x_1), \succsim_{2}\big), \succsim_1\Big)=\mathbf{x}^k.\]
It is enough to prove that whenever $\mathbf{x}\in \max\big(B^k(x_1), \succsim_{2}\big)$ and $x_1\neq x^k_1$, we have $\mathbf{x}^k\succ_1 \mathbf{x}$. By the construction of $\succsim_1$, we have $\mathbf{x}^k\succ_1\mathbf{x}$ for any $\mathbf{x}\in X\setminus \bigcup_{s\in K}\{\mathbf{x}^s\}$. Hence, it is enough to show that $\mathbf{x}^s\in \max\big(B^k(x^s_1), \succsim_{2}\big)$ implies $\mathbf{x}^k\succ_1 \mathbf{x}^s$. First, $\mathbf{x}^s\in \max\big(B^k(x^s_1), \succsim_{2}\big)$ implies that  $\mathbf{x}^s\in B^k(x^s_1)$. Moreover, by the construction of $\succsim_{2}$, $\mathbf{x}^s\in \max\big(B^k(x^s_1), \succsim_{2}\big)$ implies that $B^k(x^s_1)\subseteq B^s(x^s_1)$. That is because, if $B^k(x^s_1)\not\subseteq B^s(x^s_1)$, then $B^k(x^s_1)\setminus B^s(x^s_1)\succ_{2} \mathbf{x}^s$; i.e., $\mathbf{x}^s\not\in \max\big(B^k(x^s_1), \succsim_{2}\big)$.  Therefore, we have $\mathbf{x}^s\in B^k(x^s_1)\subseteq B^s(x^s_1)$; i.e., $\mathbf{x}^k \mathrel R_1 \mathbf{x}^s$. Hence, $\mathbf{x}^k\succ_1 \mathbf{x}^s$.

\subsection{Proof of Theorem~\ref{thm:N}.1.} \textbf{Necessity.} Let us first prove the necessity of N-SARP and Condition 1.\bigskip

\noindent\textbf{N-SARP.} Take any sequence $k_1, \ldots, k_L$ of $K$ such that for each $l\le L-1$, $\mathbf{x}^{k_l}\neq \mathbf{x}^{k_{l+1}}$ and Equation (\ref{Eq3}) holds. Note that $\mathbf{x}^{k_l}\in B^{k_{l+1}}(x^{k_{l+1}}_1)$ and $\mathbf{x}^{k_l}\neq \mathbf{x}^{k_{l+1}}$ imply that $\mathbf{x}^{k_{l+1}}\succ_2 \mathbf{x}^{k_{l}}$. Therefore, we cannot have $\mathbf{x}^{k_{L}}\succ_2 \mathbf{x}^{k_{1}}$; i.e., Equation~\eqref{Eq3} does not hold when $l=L$.\bigskip

\noindent\textbf{Condition 1.} Take any subset $S$ of $K$. For any $k\in S$, Equation~\eqref{EqN1}, $\max(B^k, \succsim_1)\in B^k(x^k_1)$, implies that 
\[\max(B^k(x^k_1), \succsim_1)\succ_1\max(B^k\setminus B^k(x^k_1), \succsim_1).\]
Consequently, we have 
\[\max(\bigcup_{k\in S} B^k(x^k_1), \succsim_1)\succ_1\max( \bigcup_{k\in S} B^k\setminus B^k(x^k_1), \succsim_1).\]
Hence, $\max(\bigcup_{k\in S} B^k(x^k_1), \succsim_1)\not\in \bigcup_{k\in S} B^k\setminus B^k(x^k_1)$; i.e., we obtain Equation~\eqref{Eq4}. 

\bigskip

\noindent\textbf{Sufficiency.} To prove the sufficiency direction, we first define the following revealed preference relation. 

\medskip
\noindent\textbf{Revealed Preference 2.} $\mathbf{x}^k \mathrel{R_2} \mathbf{x}^s$ if $\mathbf{x}^k \neq \mathbf{x}^s$ and $\mathbf{x}^s\in B^k(x^k_1)$.\medskip

By N-SARP, $R_2$ is acyclic. Let $\succsim_2$ be a complete extension of $R_2$ such that $\mathbf{x}^k\succ_2 \mathbf{x}$ for any $\mathbf{x}\in X\setminus \{\mathbf{x}^s\}_{s\in K}$. Take any $k\in K$. For any $s$ with $\mathbf{x}^k \neq \mathbf{x}^s$ and $\mathbf{x}^s\in B^k(x^k_1)$, $\exb^k\succ_2 \exb^s$. Hence, $\max\big(B^k(x^k_1), \succsim_2\big)\in X\setminus \{\mathbf{x}^s\}_{s: \mathbf{x}^k \neq \mathbf{x}^s}$. By the construction of $\succsim_2$, $\max\big(B^k(x^k_1), \succsim_2\big)=\mathbf{x}^k$.

Finally, we construct $\succsim_1$ as follows. By Condition 1, the set
\[B^{s}(x^{s}_1)\setminus \big(\bigcup_{k\in K}B^{k}\setminus B^k(x^{k}_1)\big)\]
is not empty for at least one $s\in K$. Without loss of generality, suppose that this happens for $s=1$; and let $B^{1}(x^{1}_1)\setminus \big(\bigcup_{k\in K}B^{k}\setminus B^k(x^{k}_1)\big)$ be a non-empty and $\mathbf{y}^1$ be an element of this set. Similarly, 
the set
\[B^{s}(x^{s}_1)\setminus \big(\bigcup_{k>1}B^{k}\setminus B^k(x^{k}_1)\big)\]
is non-empty for at least one $s>1$. Again, without loss of generality, let $B^{2}(x^{2}_1)\setminus \big(\bigcup_{k>1}B^{k}\setminus B^k(x^{k}_1)\big)$ be non-empty, and $\mathbf{y}^2$ be an element of this set. We follow the same procedure and obtain $\mathbf{y}^s$, for each $s\in K$, as an element of
\[B^{s}(x^{s}_1)\setminus \big(\bigcup_{k\ge s}B^{k}\setminus B^k(x^{k}_1)\big).\]

Let $\succsim_1$ be a preference relation such that $\textbf{y}^s\succ_1\textbf{y}^{s+1}$ for each $s<K$ and $\textbf{y}^K\succ_1\mathbf{x}$ for any $\text{x}\in X\setminus\{\textbf{y}^k\}_{k\in K}$. Since $\textbf{y}^s\in B^s(x^s_1)$ and $\textbf{y}^s\succ_1\mathbf{x}$ for any $\text{x}\in X\setminus\{\textbf{y}^k\}_{k<s}$, 
\[\max\big(B^s, \succsim_1\big)\in \{\textbf{y}^1, \ldots, \textbf{y}^s\}.\]
Suppose for some $t\le s$, 
\[\max\big(B^s, \succsim_1\big)=\textbf{y}^t.\]
By the construction, $\textbf{y}^t\in B^{t}(x^{t}_1)\setminus \big(\bigcup_{k\ge t}B^{k}\setminus B^k(x^{k}_1)\big)$, which implies 
\[\textbf{y}^t\in B^{t}(x^{t}_1)\setminus \big(B^{s}\setminus B^s(x^{s}_1)\big).\] Since $\mathbf{y}^t\in B^s$, $\textbf{y}^t\in B^{t}(x^{t}_1)\setminus \big(B^{s}\setminus B^s(x^{s}_1)\big)$ implies that  $\mathbf{y}^t\in B^s(x^{s}_1)$, i.e.,  
\[\mathbf{y}^t=\max\big(B^s, \succsim_1\big)\in B^s(x^s_1).\]

\subsection{Proof of Theorem~\ref{thm:N}.2.} It is immediate from standard results that rationalization via Equation \eqref{EqNN2} is equivalent to N-SARP. Hence, N-SARP is necessary. Moreover, N-SARP is necessary since Equation \eqref{EqNN1} can be always satisfied by setting $\succsim_1$ to be indifferent among all elements of $X$.

\subsection{Proof of Theorem~\ref{thm:N}.3.} It is immediate from standard results that rationalization via Equation \eqref{EqSNN2} is equivalent to N-SARP and rationalization via Equation \eqref{EqSNN1} is equivalent to NN-SARP. 

\subsection{Proof of Theorems~\ref{thm:sopgen} and~\ref{thm:naigen}}

Theorem \ref{thm:sopgen}.3 and Theorem \ref{thm:naigen}.2-3 immediately follow from the following two observations. First, from standard results, Equation (\ref{Eq38}) is characterized by $t$-SARP. Second, Equation (\ref{Eq36}) has no testable implications as we can set $\succsim_t$ to be indifferent between all alternatives of $X$. Theorem \ref{thm:sopgen}.2 follows from Theorem \ref{thm:naigen}.1. Hence, we only prove Theorem \ref{thm:sopgen}.1 and Theorem \ref{thm:naigen}.1.  

\subsubsection{Proof of Theorem \ref{thm:sopgen}.1.} First we show necessity. Suppose
that the dataset $\D$ is sophisticated rationalizable by means of the preference
relations $\succsim_t$, $1\leq t\leq T$. To prove the necessity of
Condition 2, recall the revealed preference relation $R_1$ defined in the proof of
Theorem~\ref{thm:sophdistinct}: $\mathbf{x}^k \mathrel R_1 \mathbf{x}^s$ if
$\mathbf{x}^s\in B^k(x^s_1) \subseteq B^s(x^s_1)$. We shall prove that $R_1$ is
acyclic. Take any $k, s$ with $\mathbf{x}^k \mathrel R_1 \mathbf{x}^s$. To prove acyclicity, we will show that $\mathbf{x}^k\succ_1\mathbf{x}^s.$ 

Since $\mathbf{x}^s=\max(M^s_1, \succsim_1)$, we have 
\[\mathbf{x}^s\in M^s_1=\bigcup_{x_{1}\in B^k_{1}}\max\big(M^k_{2}(x_1), \succsim_{2}\big),\]
consequently, $\mathbf{x}^s\in M^s_2(x^s_1)$.
Since $\mathbf{x}^s\in B^k(x^s_1) \subseteq B^s(x^s_1)$, we then have 
$\mathbf{x}^s\in B^k(\mathbf{x}^s_t) \subseteq B^s(\mathbf{x}^s_t)$ for each
$t<T$. Hence, $\mathbf{x}^s\in M^s_2(x^s_1)$ implies $\mathbf{x}^s\in
M^k_2(x^s_1)$.\footnote{Observe that working with a strong rationalization is key
  here, as otherwise there an agent who moves later could choose a different optimal
  response in observation $k$ from what they do in $s$.} Finally, since $\mathbf{x}^k=\max (M^k_1, \succsim_1)$, we have $\mathbf{x}^k\succ_1\mathbf{x}^s.$

We now prove the sufficiency of Condition 2. Let $\hat X_2=\prod^T_{t=2} X_t$. Then $X=X_1\times \hat X_2 $. By Theorem \ref{thm:sophdistinct}, there are preferences $\succsim_1$ and $\succsim_2$ over $X$ such that 
for each observation $k\in K$,
\begin{equation}\label{eq28}
\max\Big(\bigcup_{x_1\in B^k_{1}}\max\big(B^k(x_1), \succsim_2\big), \succsim_1\Big)=\mathbf{x}^k.\end{equation}

Let $\succsim_t=\succsim_2$ for each $t\ge 3$. We can show that Equation (\ref{Eq26}) is satisfied. Take any $k\in K$.

\medskip
\noindent\textbf{Fact 1.} For any $A, B\subseteq X$, \[\max(A\cup B, \succsim_2)=\max\big(\max(A, \succsim_2)\cup \max(B, \succsim_2), \succsim_2\big).\]

Fact 1 is the usual path independence property of rational choice. We provide a proof
for completeness: Take any $\mathbf{x}\in \max(A\cup B, \succsim_2)$ and $\mathbf{x}\in A$. Then $\mathbf{x}\in \max(A, \succsim_2)$. Since $\mathbf{x}\succsim_2 \max(B, \succsim_2)$, 
$\mathbf{x}\in \max\big(\max(A, \succsim_2)\cup \max(B, \succsim_2),
\succsim_2\big)$. Conversely, consider any $\mathbf{x}\in \max\big(\max(A, \succsim_2)\cup \max(B, \succsim_2), \succsim_2\big)$. In other words, $\mathbf{x}\succsim_2 \max(A, \succsim_2)$ and $\mathbf{x}\succsim_2 \max(B, \succsim_2)$. Hence, $\mathbf{x}\in \max(A\cup B, \succsim_2)$.

\medskip
\noindent\textbf{Fact 2.} $\max(B^k(x_1), \succsim_2)=\max(M^k_2(x_1), \succsim_2)$.\medskip

We proceed to prove Fact 2.  By Fact 1, Equation (\ref{Eq27}) and
$\succsim_t=\succsim_2$ implies that

\begin{equation}
\max(M^k_t(\mathbf{x}_{t-1}), \succsim_2)=\max\Big(\bigcup_{x_{t}\in B^k_{t}(\mathbf{x}_{t-1})} M^k_{t+1}(\mathbf{x}_{t}), \succsim_2\Big).
\end{equation}

By repeatedly using Fact 1, the above equation implies 

\begin{equation}
\max(M^k_1(x_1), \succsim_2)=\max\Big(\bigcup_{x_2, \ldots, x_t: \mathbf{x}\in B^k(x_1)} M^k_{T}(\mathbf{x}_{T-1}), \succsim_2\Big)=\max(B^k(x_1), \succsim_2),
\end{equation} which establishes Fact 2.

Finally, by Facts 1 and 2, Equation (\ref{eq28}) implies that

\begin{equation}\label{eq29}
\max\Big(\bigcup_{x_1\in B^k_{1}}\max\big(M^k_2(x_1), \succsim_2\big), \succsim_1\Big)=\mathbf{x}^k,\end{equation}
which implies Equation (\ref{Eq26}).

\subsubsection{Proof of Theorem \ref{thm:naigen}.1.} We first prove the necessity of $T$-SARP and Condition 5. 

\medskip
\noindent\textbf{$T$-SARP.} Consider the following revealed preference relation:
$\mathbf{x}^k \mathrel R_T \mathbf{x}^s$ if $\mathbf{x}^k\neq \mathbf{x}^s$ and $\mathbf{x}^{s}\in B^{k}(\mathbf{x}^{k}_{T-1})$.
By Equation (\ref{Eq31}),  $\mathbf{x}^k \mathrel R_T \mathbf{x}^s$ implies $\mathbf{x}^k \succ_T \mathbf{x}^s$. Hence, $R_T$ is acyclic; i.e., $T$-SARP is satisfied.

\medskip
\noindent\textbf{Condition 5.} Take any subset $S$ of $K$ and $t\le T$. By Equation (\ref{Eq30}), we have $\max(B^{k}(\mathbf{x}^k_{t-1}), \succsim_t)\in  B^k(\mathbf{x}^k_t)$, which implies 
\[\max(B^k(\mathbf{x}^k_t), \succsim_t)\succ_t \max(B^{k}(\mathbf{x}^k_{t-1})\setminus B^k(\mathbf{x}^k_t), \succsim_t).\]
By adding across $k\in S$, we obtain
\[\max(\bigcup_{k\in S} B^k(\mathbf{x}^k_t), \succsim_t)\succ_t \max(\bigcup_{k\in S} [B^{k}(\mathbf{x}^k_{t-1})\setminus B^k(\mathbf{x}^k_t)], \succsim_t).\]
Hence, \[\bigcup_{k\in S} B^k(\mathbf{x}^k_t), \succsim_t)\not\subseteq \bigcup_{k\in S} [B^{k}(\mathbf{x}^k_{t-1})\setminus B^k(\mathbf{x}^k_t)].\]

\noindent\textbf{Sufficiency.} We now prove the sufficiency of $T$-SARP and Condition 5. By $T$-SARP, $R_T$ is acyclic. Let $\succsim_T$ be a complete extension of $R_T$ such that $\mathbf{x}^k\succ_T \mathbf{x}$ for any $\mathbf{x}\in X\setminus \{\mathbf{x}^s\}_{s\in K}$. Take any $k\in K$. For any $s$ with $\mathbf{x}^k \neq \mathbf{x}^s$ and $\mathbf{x}^s\in B^k(\mathbf{x}^k_{T-1})$, $\exb^k\succ_T \exb^s$. Hence, $\max\big(B^k(\mathbf{x}^k_{T-1}), \succsim_T\big)\in X\setminus \{\mathbf{x}^s\}_{s: \mathbf{x}^k \neq \mathbf{x}^s}$. By the construction of $\succsim_T$, $\max\big(B^k(\mathbf{x}^k_{T-1}), \succsim_T\big)=\mathbf{x}^k$.

Finally, we construct $\succsim_t$ as follows. By Condition 5, the set
\[B^{s}(\mathbf{x}^s_{t})\setminus \big(\bigcup_{k\in K}B^{k}(\mathbf{x}^k_{t-1})\setminus B^k(\mathbf{x}^k_{t})\big)\]
is not empty for at least one $s\in K$. Without loss of generality, suppose that this
happens for $s=1$. So  $B^{1}(\mathbf{x}^1_{t})\setminus \big(\bigcup_{k\in
  K}B^{k}(\mathbf{x}^k_{t-1})\setminus B^k(\mathbf{x}^k_{t})\big)$ is non-empty, and
we may choose $\mathbf{y}^1$ to be an element of this set. Similarly, 
the set
\[B^{s}(\mathbf{x}^s_{t})\setminus \big(\bigcup_{k>1}B^{k}(\mathbf{x}^k_{t-1})\setminus B^k(\mathbf{x}^k_{t})\big)\]
is non-empty for at least one $s\neq 1$. Again, without loss of generality, let $B^{2}(\mathbf{x}^2_{t})\setminus \big(\bigcup_{k>1}B^{k}(\mathbf{x}^k_{t-1})\setminus B^k(\mathbf{x}^k_{t})\big)$ be non-empty, and $\mathbf{y}^2$ be an element of this set. We follow the same procedure and obtain $\mathbf{y}^s$, for each $s\in K$, as an element of
\[B^{s}(\mathbf{x}^s_{t})\setminus \big(\bigcup_{k\ge s}B^{k}(\mathbf{x}^k_{t-1})\setminus B^k(\mathbf{x}^k_{t})\big).\]

Let $\succsim_t$ be a preference relation such that
$\textbf{y}^s\succ_t\textbf{y}^{s+1}$ for each $1\leq s<K$, and
$\textbf{y}^K\succ_t\mathbf{x}$ for any $\text{x}\in X\setminus\{\textbf{y}^k\}_{k\in
  K}$. We shall prove that, for any $s$, $\max(B^s(\mathbf{x}^s_{t-1}),
\succsim_t)\in B^s(\mathbf{x}^s_t)$, which guarantees that the period $t$ agent with
preferences $\succsim_t$ chooses $x^s_t$.

Since $\textbf{y}^s\in B^s(\mathbf{x}^s_{t})\subseteq B^s(\mathbf{x}^s_{t-1})$ and $\textbf{y}^s\succ_t\mathbf{x}$ for any $\mathbf{x}\in X\setminus\{\textbf{y}^k\}_{k\le s}$, 
\[\max\big(B^s(\mathbf{x}^s_{t-1}), \succsim_t\big)\in \{\textbf{y}^1, \ldots, \textbf{y}^s\}.\]
Suppose for some $s'\le s$, 
\[\max\big(B^s(\mathbf{x}^s_{t-1}), \succsim_t\big)=\textbf{y}^{s'}.\]
By the construction, $\textbf{y}^{s'}\in B^{s'}(\mathbf{x}^{s'}_{t})\setminus \big(\bigcup_{k\ge s'}B^{k}(\mathbf{x}^k_{t-1})\setminus B^k(\mathbf{x}^k_t)\big)$, which implies 
\[\textbf{y}^{s'}\in B^{s'}(\mathbf{x}^{s'}_{t})\setminus \big(B^{s}(\mathbf{x}^s_{t-1})\setminus B^s(\mathbf{x}^s_t)\big).\] Since $\textbf{y}^{s'}\in B^{s'}(\mathbf{x}^{s'}_t)$ and $\mathbf{y}^{s'}\in B^{s}(\mathbf{x}^s_{t-1})$, we must have $\textbf{y}^{s'}\in B^s(\mathbf{x}^s_t)$, i.e., \[\max\big(B^s(\mathbf{x}^s_{t-1}), \succsim_t\big)\in B^s(\mathbf{x}^s_t).\]

\bibliographystyle{ecta}
\bibliography{soph}

\appendix

\section{FOCs rationalization with many observations}\label{sec:focsappendix}

\begin{theorem}\label{thm:manyobs}
    There is a dataset of arbitrarily large size that is strong FOCs rationalizable by the sophisticaed quasi-hyperbolic model, but not equilibrium rationalizable by the sophisticated hyperbolic model.
\end{theorem}

\begin{proof} Fix two numbers, $\ul\ta=3/100$ and $\bar\ta=90/100$. Consider a dataset with $K+2$ observations $(x^k,p^k)$, $1\leq k\leq K+2$. These are as follows:
\begin{enumerate}
\item $x^k_2=x^k_3=0.1$ and $p^k_2=\frac{1}{(0.8)^2}$ for $1\le k\le K$
\item $x^k_1= \bar\ta \frac{k}{K}$ and $p^k_1 = \frac{1+.8^{2}}{.99 (1.8).8^4} ( 1 - (\bar\ta \frac{k}{K})^2 )$ for $1\le k\le K$
\item $x^{K+1}_1=x^{K+1}_2=x^{K+1}_3=0.2$, $p^{K+1}_2=\frac{1}{(0.8)^2}$, and $p^{K+1}_1=\frac{5125}{2304}$.
\item $x^{K+2}_1=0.0757$, $x^{K+2}_2=0.0367$, $x^{K+2}_3=0.4688$, $p^{K+2}_2=2$, and $p^{K+2}_1=3.0884$.
\end{enumerate}

As we will see below the data is strong FOCs rationalizable with $(\hat{u},\beta,\delta)$ where $\hat{u}(x)=x-\frac{x^3}{3}$ when $x\in [\underline{\theta}, \overline{\theta}]$ and $\beta=\delta=0.8$.

Suppose now that $(u,\beta,\delta)$ equilibrium rationalizes the data. By a positive affine transformation, we may normalize $u$ so that $u(\ul\ta)=\ul\ta-\ul\ta^3/3$ and  $u'(0.1)=.99$. By Equation~\eqref{eq:foc2}, then, $x^k_2=x^k_3=0.1$ and $p^k_2=\frac{1}{(0.8)^2}$ for any $k\le K$ imply that $1=p^k_2\beta\da$, so $\beta\delta = .8^2$.

Consider now observation $K+1$, where we also have $x^{K+1}_2=x^{K+1}_3$ and $p^{K+1}_2=\frac{1}{(0.8)^2}$. Then Equation~\eqref{eq:foc2} implies that $A=(\beta \delta p^{K+1}_2)^{-1} =1$ for this observation, and therefore $g(x_2)=x_2$ for Agent 2's strategy in the game defined by observation $K+1$. As a consequence  we have that $g'(x_2)=1$. Then by Equation \eqref{eq:foc1}, and using the numbers in observation $K+1$, we obtain that 
\[1=\frac{u'(0.2)}{u'(0.2)}=\frac{\delta\, \frac{5125}{2304}}{\frac{1}{.8^2}}\frac{1+\frac{\beta}{.8^2}}{1+\frac{1}{.8^2}}=\frac{5125}{2304}\frac{\delta+1}{0.8^{-2}+0.8^{-4}}=\frac{\delta+1}{1.8}.\]
Hence, $\beta=\da=0.8$

We turn next to the observations $k=1,\ldots,K$. Again the fact that consumption in periods 2 and 3 are the same, and that $p^k_2\beta\delta=1$ means that $g'(x_2)=1$ for such observations. Then Equation \eqref{eq:foc1} implies
\begin{align*}
u'(x^k_1) & = u'(0.1)\delta\frac{p^k_1}{p^k_2}\frac{1+\beta p^k_2}{1+p^k_2}=\frac{u'(0.1)\beta\delta (1+\delta)}{1+.8^{-2}} p^k_1 \\
& = \frac{u'(0.1)0.8^2 (1.8)}{1+.8^{-2}} [ \frac{1+.8^{2}}{.99 (1.8).8^4} ( 1 - (\bar\ta k/K)^2 )  ] \\
& =   1 - (x^k_1)^2 
\end{align*}

Let $\hat u(x) = x-x^3/3$. The calculations above mean that $u'(x) = \hat u'(x)$ for all $x\in \{z_1,\ldots,z_K\}$, with $0<z_1<\ldots < z_K= \bar \ta$ and $z_{k+1}-z_k<1/K$. Then for any $x\in [z_1,\bar \ta]$ we have that
\[\begin{split}
u'(x) - \hat u'(x) \leq u'(z_k) - \hat u'(z_{k+1}) = \hat u'(z_k) - \hat u'(z_{k+1}) = z_{k+1}^2-z_k^2  \\
= (z_{k+1}-z_k) (z_{k+1} + z_k)< \frac{2}{K},
\end{split}\] where we have chosen $k$ so that $z_k\leq x\leq z_{k+1}$, and used the concavity of $u'$ and $\hat u'$. Similarly, $\hat u'(x) - u'(x)< 2/K$. 

In consequence we have $\sup\{u'(x)-\hat u'(x):x\in [\ul\ta,\bar\ta] \}< 2/K$ and therefore $\abs{u(x)-\hat u(x)}< 2/K$ for all $x\in [\ul\ta,\bar\ta]$ as $u(\ul\ta)=\hat u(\ul\ta)$.

Now consider the last observation $K+2$. Observe that in, the budget set for observation $K+2$, all affordable bundles involve quantities that are smaller than $\bar \ta$.

Here $u$ and $\hat u$ both match the first order conditions \eqref{eq:foc1} and \eqref{eq:foc2}. For this observation we obtain $A^* = (p^{K+2}_2\beta \da)^{-1}=0.78125$. Let $\hat g$ and $\hat f$ denote the functions $f$ and $g$ corresponding to utility function $\hat u$, $A^*$, $p_2=p^{K+2}_2$, $p_1=p^{K+1}_1$ and $m=m^{K+1}$. We use $f$ and $g$ to denote these functions for utility $u$ and the parameters of the $K+2$ budget.

To check the second-order condition for $\hat u$, recall Equation \eqref{eq:SOCn}. Since \[\hat g'_k(x_2)=\frac{A^*x_2}{\sqrt{1-A^*+A^* (x_2)^2}}=0.0611\text{ and }\hat g''_k(x_2)=\frac{A^*(1-A^*)}{(1-A^*+A^* (x_2)^2)^\frac{3}{2}}=1.658,\]
we have
\begin{eqnarray*}\hat g''_k(x_2) \hat u'(x_2)\frac{\delta(1-\beta)}{\hat g'(x_2)+p^{K+2}_2}\!\!&\!\!=\!\!&\!\!0.128\\
\!\!&\!\!>\!\!&\!\!|\hat u''(x_1)(\hat f'(x_2))^2+\beta\delta \hat u''(x_2)+\beta \delta^2 \hat u''(x_3) (\hat g'(x_2))^2|=0.116.\end{eqnarray*}
So the second-order condition, Equation \eqref{eq:SOCn}, is violated.


This means that there is $\Delta>0$ so that $\hat u(x^{K+2}_1)+\beta\da \hat u(x^{K+2}_2)+\beta\da^2 \hat u(x^{K+2}_3)+\Delta<M$, where $M$ is the optimal utility for Agent 1's maximization program. Let $(\hat x_1,\hat x_2,\hat x_3)$ achieve utility $M$. So $\hat x_3=\hat g(\hat x_2)$, $\hat x_1=\hat f(\hat x_2)$, and $M=\hat u(\hat x_1)+\beta\da \hat u(\hat x_2)+\beta\da^2 \hat u(\hat x_3)$.

First, choose $\ep>0$ and $K'>\frac{1}{\ul\ta}$ so that
\begin{equation}\label{eq:bound1}
\ep + 2/K' + 0.8^2 (2/K') + 0.8^3 (\ep + 2/K')< \Delta/3
\end{equation}

Next, choose $K\geq K'$ so that
\begin{equation}\label{eq:bound2}
(A^*+1)\frac{2}{K}<\ep^2.
\end{equation}

Consider the utility under utility function  $u$ when Agent 1 chooses consumption $\hat x_2$, the optimal choice under utility function $\hat u$. Let $x_3=g(\hat x_2)$ and $x_1=f(\hat x_2)$.

Then we have $\hat u'(\hat x_3)= A^* \hat u'(\hat x_2)$ and $u'(x_3)= A^* u'(\hat x_2)$. So\[
A^*(\hat u'(\hat x_2) - u'(\hat x_2)) + u'(x_3) - \hat u'(x_3) = \hat u'(\hat x_3) - \hat u'(x_3),\] which implies that
\[\begin{split}
A^*(2/K) + 2/K\geq 
A^*\abs{\hat u'(\hat x_2) - u'(\hat x_2))} + \abs{u'(x_3) - \hat u'(x_3)}
\geq \abs{\hat u'(\hat x_3) - \hat u'(x_3)} \\
 = \abs{\hat x^2_3 - x^2_3} = \abs{\hat x_3 - x_3}(\hat x_3 + x_3).
\end{split}
\]
Now we claim that $\abs{\hat x_3 - x_3}<\ep$. There are two cases. If $\hat x_3 + x_3<\ep$ then we are done because the difference between two positive numbers is smaller than their sum. If, instead, $\hat x_3 + x_3\geq \ep$ then~\eqref{eq:bound2} and our choice of $K$ implies that $\ep > \frac{1}{\ep}(A^*(2/K) + 2/K)$ so we also conclude that $\abs{\hat x_3 - x_3}<\ep$.

As a consequence, we obtain that 
\[
\abs{\hat u(\hat x_3) - u(g(\hat x_2))}
\leq
\abs{\hat u(\hat x_3) - \hat u(x_3)} + \abs{\hat u(x_3) - u(x_3)}
\leq \abs{\hat x_3 - x_3} + 2/K<\ep +2/K,\] as $\hat u'$ is monotone decreasing and $\hat u'(0)=1$. 

Finally, 
\[
\abs{\hat f(\hat x_2) - f(\hat x_2)} = \frac{1}{p^{K+2}_1}\abs{\hat g(\hat x_2) - g(\hat x_2)}<\abs{\hat x_3-x_3}<\ep,\] as $p^{K+2}_1=3.0884$. So, by the same reasoning as above, we have that $\abs{\hat u(\hat x_1) - u(f(\hat x_2))}<\ep+2/K$.

Putting everything together we conclude, by~\eqref{eq:bound1}, that
\[
[u(f(\hat x_2)) + \beta\da u(\hat x_2) + \beta\da^2 u(g(\hat x_2))]
-[u(x^{K+2}_1)+\beta\da u(x^{K+2}_2)+\beta\da^2 u(x^{K+2}_3)]>\Delta/3
\]
since the above is equal to the sum of $C1$, $C2$, and $C3$, where 
\begin{align*}
C1 & = u(f(\hat x_2)) - \hat u(\hat x_1)
+\beta\da [u(\hat x_2) - \hat u(\hat x_2)]
+\beta\da^2 [u(g(\hat x_2)) - \hat u(\hat x_3)] >-\Delta/3 \\
C2 & =
\hat u(\hat x_1) - \hat u (x^{K+2}_1)
+\beta\da [\hat u(\hat x_2) - \hat u (x^{K+2}_2)]
+\beta\da^2 [\hat u(\hat x_2) - \hat u (x^{K+1}_2)]> \Delta \\
C3 & = \hat u(x^{K+2}_1) - u(x^{K+2}_1) \\
& +\beta\da [\hat u(x^{K+2}_2) - u(x^{K+2}_2)]
+\beta\da^2 [\hat u(x^{K+2}_3) - u(x^{K+2}_3)]  >-\Delta/3 \\
\end{align*}
where the inequalities for $C1$ and $C2$ follow from \eqref{eq:bound1} and the bounds we have established above.

Finally, we note that the optimal choices $(\hat x_1,\hat x_2,\hat x_3)$ for $\hat u$, achieving utility $M$, can be calculated to be $\hat x_1=0.06$, $\hat x_2=0.06$, and $\hat x_3=0.4707$, all of which are in $[\ul\ta,\bar\ta]$. We conclude then that $x^{K+2}_2$ is not an optimal choice for Agent 2 when the utility is $(u,\beta,\da)$.

\end{proof}

\section{Necessary Conditions for Sophisticated Rationalization}\label{sec:appB}

\begin{namedaxiom}[Condition 3]
There is no sequence $k_1, \ldots, k_L$ of $K$ such that for each $l\le L$, $\mathbf{x}^{k_l}\neq \mathbf{x}^{k_{l+1}}$ and 
\begin{equation}\label{Eq6}
\mathbf{x}^{k_l}\in B^{k_{l+1}}(x^{k_{l}}_1)\subseteq \bigcup_{t:\mathbf{x}^t\in  B^{k_l}(x^{k_l}_1)} B^t(x^t_1).\end{equation}
\end{namedaxiom}

\medskip
\noindent\textbf{Remark 2.} Condition 3 implies Condition 2 as Equation (\ref{Eq7}) implies Equation (\ref{Eq6}). When $x^k_1\neq x^s_1$ for any $k, s\in K$, Conditions 2 and 3 are equivalent. \medskip

\medskip
\noindent\textbf{Revealed Preference 3.} $\mathbf{x}^k R_1 \mathbf{x}^s$ if $\mathbf{x}^k \neq \mathbf{x}^s$ and
\[\mathbf{x}^s\in B^k(x^s_1)\subseteq \bigcup_{t:\mathbf{x}^t\in  B^s(x^s_1)} B^t(x^t_1).\]

Under Condition 3, $R_1$ is acyclic. We write $\mathbf{x}^{k} \,\text{tran}(R_1)\, \mathbf{x}^{s}$ if there is $k_1, \ldots, k_L$ such that $x^{k_1}=x^k$, $x^{k_L}=x^s$, and $x^{k_l} R_1 x^{k_{l+1}}$ for each $l<L$.\medskip

\begin{namedaxiom}[Condition 4] There is no sequence $k_1, \ldots, k_L, s_1, \ldots, s_L$ of $K$ such that for each $l\le L$,  
\begin{equation}\label{eq20}
\mathbf{x}^{s_l} \text{tran}(R_1) \mathbf{x}^{k_l}\text{ and }\mathbf{x}^{s_l}\in B^{k_{l}}(x^{s_{l}}_1)\text{ and }x^{s_l}_1\neq x^{k_l}_1, 
\end{equation}
and
\begin{equation}\label{eq21}
\bigcup_{l} \left[B^{k_{l}}(x^{s_{l}}_1)\setminus\Big[\bigcup_{t:\mathbf{x}^t\in  B^{s_l}(x^{s_l}_1)} B^t(x^t_1)\Big]\right]\subseteq \bigcup_{s\in S} B^{s}(x^{s}_1),
\end{equation}
for any $s\in S$,
\begin{equation}\label{eq22}
\mathbf{x}^{s_l} R_2 \mathbf{x}^s\text{ or }\mathbf{x}^{s_l}= \mathbf{x}^s\text{ for some $l$.}
\end{equation}
\end{namedaxiom}

\begin{proposition}\label{prop6} If $\D$ is sophisticated rationalizable, then it satisfies N-SARP and Conditions 3-4.  \end{proposition}

\subsection{\textbf{Proof of Proposition \ref{prop6}}} \textbf{N-SARP.} To prove the necessity of N-SARP, it is now enough to show the acyclicity of $R_2$. Take any $k, s$ such that $\mathbf{x}^k R_2 \mathbf{x}^s$; i.e., $\mathbf{x}^k\neq \mathbf{x}^s$ and $\mathbf{x}^s\in B^k(x^k_1)$. Since 
\[\max\Big(\bigcup_{x_1\in B^k_{1}}\max\big(B^k(x_1), \succsim_{2}\big), \succsim_1\Big)=\mathbf{x}^k,\]
we have
\[\max\Big(\max\big(B^k(x^k_1), \succsim_{2}\big), \succsim_1\Big)=\mathbf{x}^k\text{ and }\mathbf{x}^k\in \max\big(B^k(x^k_1), \succsim_{2}\big).\]
First, $\mathbf{x}^k\in \max\big(B^k(x^k_1), \succsim_{2}\big)$ and $\mathbf{x}^s\in B^k(x^k_1)$ imply $\mathbf{x}^k\succsim_{2}\mathbf{x}^s$. Moreover, if $\mathbf{x}^k\sim_{2}\mathbf{x}^s$, $\max\Big(\max\big(B^k(x^k_1), \succsim_{2}\big), \succsim_1\Big)=\mathbf{x}^k$ implies $\mathbf{x}^k\succ_{1}\mathbf{x}^s$. In other words, $\mathbf{x}^k R_2 \mathbf{x}^s$ implies that either $\mathbf{x}^k\succ_{2}\mathbf{x}^s$ or $\mathbf{x}^k\sim_{2}\mathbf{x}^s$ and $\mathbf{x}^k\succ_{1}\mathbf{x}^s$. Therefore, $R_2$ is acyclic.

\smallskip
\noindent\textbf{Condition 3.} To prove the necessity of Condition 3, it is enough to prove the acyclicity of $R_1$. Take any $k, s$ such that $\mathbf{x}^k R_1 \mathbf{x}^s$. Since 
\[\max\Big(\bigcup_{x_1\in B^s_{1}}\max\big(B^s(x_1), \succsim_{2}\big), \succsim_1\Big)=\mathbf{x}^s,\]
we have
\[\max\Big(\max\big(B^s(x^s_1), \succsim_{2}\big), \succsim_1\Big)=\mathbf{x}^s;\]
which implies
\[\mathbf{x}^s\in \max\big(B^s(x^s_1), \succsim_{2}\big).\]
Moreover, since 
\[\max\Big(\big\{\max\big(B^k(x_1), \succsim_{2}\big)\big\}_{x_1\in B^k_{1}}, \succsim_1\Big)=\mathbf{x}^k\text{ and }x^s_1\in B^k_1,\]
we have
\[\mathbf{x}^k\succ_1 \max\Big(\max\big(B^k(x^s_1), \succsim_{2}\big), \succsim_1\Big)\text{ or }\max\Big(\max\big(B^k(x^s_1), \succsim_{2}\big), \succsim_1\Big)=\mathbf{x}^k.\]
If $\mathbf{x}^s\in \max\big(B^k(x^s_1), \succsim_{2}\big)$, then we have $\mathbf{x}^k\succ_1\mathbf{x}^s$. Suppose now  $\mathbf{x}^s\not\in \max\big(B^k(x^s_1), \succsim_{2}\big)$. Since 
\[\mathbf{x}^s\in \max\big(B^s(x^s_1), \succsim_{2}\big),\]
we have 
\[\max\big(B^k(x^s_1)\setminus B^s(x^s_1), \succsim_{2}\big)\succ_{2} \mathbf{x}^s.\]

However, \[B^k(x^s_1)\subseteq \bigcup_{t:\mathbf{x}^t\in  B^s(x^s_1)} B^t(x^t_1)\]
and the definition of $R_2$
imply that 
\[B^k(x^s_1)\setminus B^s(x^s_1)\subseteq \bigcup_{\mathbf{x}^s R_2 \mathbf{x}^t} B^t(x^t_1).\]
Therefore, 
\[\max(\bigcup_{\mathbf{x}^s R_2 \mathbf{x}^t} B^t(x^t_1), \succsim_{2})\succsim_{2}\max\big(B^k(x^s_1)\setminus B^s(x^s_1), \succsim_{2}\big)\succ_{2} \mathbf{x}^s.\]

However, by the proof of the acyclicity of $R_2$, $\mathbf{x}^s R_2 \mathbf{x}^t$ implies that $\mathbf{x}^s \succsim_{2} \mathbf{x}^t$. Moreover, since 
\[\max\Big(\max\big(B^t(x^t_1), \succsim_{2}\big), \succsim_1\Big)=\mathbf{x}^t,\]
we have
\[\mathbf{x}^t\in \max\big(B^t(x^t_1), \succsim_{2}\big).\]
Therefore, we have 
\[\mathbf{x}^s\succsim_{2} \max(\bigcup_{\mathbf{x}^s R_2 \mathbf{x}^t} B^t(x^t_1), \succsim_{2}),\] 
a contradiction.

\smallskip
\noindent\textbf{Condition 4.} To prove the necessity of Condition 4, by way of contradiction, suppose there is a sequence $k_1, \ldots, k_L, s_1, \ldots, s_L$ of $K$ such that for each $l\le L$,  
\begin{equation}\label{eq23}
\mathbf{x}^{s_l} \text{tran}(R_1) \mathbf{x}^{k_l}\text{ and }\mathbf{x}^{s_l}\in B^{k_{l}}(x^{s_{l}}_1)\text{ and }x^{s_l}_1\neq x^{k_l}_1, 
\end{equation}
and
\begin{equation}\label{eq24}
\bigcup_{l} \left[B^{k_{l}}(x^{s_{l}}_1)\setminus\Big[\bigcup_{t:\mathbf{x}^t\in  B^{s_l}(x^{s_l}_1)} B^t(x^t_1)\Big]\right]\subseteq \bigcup_{s\in S} B^{s}(x^{s}_1),
\end{equation}
for any $s\in S$,
\begin{equation}\label{eq25}
\mathbf{x}^{s_l} R_2 \mathbf{x}^s\text{ or }\mathbf{x}^{s_l}= \mathbf{x}^s\text{ for some $l$.}
\end{equation}

First, $\mathbf{x}^{s_l} \text{tran}(R_1) \mathbf{x}^{k_l}$ implies that $\mathbf{x}^{s_l} \succ_1 \mathbf{x}^{k_l}$. Then, since $\mathbf{x}^{s_l} \succ_1 \mathbf{x}^{k_l}$ and 
\[\max\Big(\bigcup_{x_1\in B^{k_l}_{1}}\max\big(B^{k_l}(x_1), \succsim_{2}\big), \succsim_1\Big)=\mathbf{x}^{k_l},\]
we need to have $\mathbf{x}^{s_l}\not\in \max\big(B^{k_l}(x^{s_l}_1), \succsim_{2}\big)$. Since, $\mathbf{x}^{s_l}\in B^{k_{l}}(x^{s_{l}}_1)\text{ and }x^{s_l}_1\neq x^{k_l}_1$ imply that there is some $\mathbf{x}\in B^{k_{l}}(x^{s_{l}}_1)\setminus\Big[\bigcup_{t:\mathbf{x}^t\in  B^{s_l}(x^{s_l}_1)} B^t(x^t_1)\Big]$ such that $\mathbf{x}\succ_2 \mathbf{x}^{s_l}$. By Equation (\ref{eq24}), there is $s$ such that $\mathbf{x}\in B^s(x^s_1)$. Hence, $\mathbf{x}^s\succsim_2 \mathbf{x}$. In other words, $\mathbf{x}^s\succ_2 \mathbf{x}^{s_l}$. Let $\mathbf{x}^{s^*}\in \max(\{\mathbf{x}^s\}_{s\in S}, \succsim_2)$. Then $\mathbf{x}^{s^*}\succ_2 \mathbf{x}^{s_l}$ for each $l$. In other words, there is no $l$ that satisfies Equation (\ref{eq25}).

\end{document}